\documentclass[a4paper,10pt,notitlepage]{article}
\usepackage{amsmath}
\usepackage{amssymb}
\usepackage{bbm}
\usepackage{xstring}
\usepackage[english]{babel}
\usepackage{hyperref}
\usepackage[amsthm,thref,hyperref,thmmarks]{ntheorem}
\usepackage{caption}
\usepackage{subcaption}
\usepackage[boxed]{algorithm2e}
\usepackage{xcolor} 
\usepackage{graphicx}
\usepackage[a4paper,left=20mm,right=20mm,top=20mm,bottom=20mm]{geometry}

\newcommand{\refP}[1]{%
	\def\InputString{#1}%
	\IfBeginWith{\InputString}{Equation}{%
		(\ref{#1})}{%
	\IfBeginWith{\InputString}{Section}{%
		Section \ref{#1}}{%
	\IfBeginWith{\InputString}{Subsection}{%
		Subsection \ref{#1}}{%
	\IfBeginWith{\InputString}{Chapter}{%
		Chapter \ref{#1}}{%
	\IfBeginWith{\InputString}{Subsubsection}{%
		Subsubsection \ref{#1}}{%
	\IfBeginWith{\InputString}{Problem}{%
		(\ref{#1})}{%
	\IfBeginWith{\InputString}{Property}{%
		property (\ref{#1})}{%
	\IfBeginWith{\InputString}{Algorithm}{%
		Algorithm \ref{#1}}{%
	\IfBeginWith{\InputString}{Figure}{%
		Figure \ref{#1}}{%
	\IfBeginWith{\InputString}{Table}{%
		Table \ref{#1}}{%
	\IfBeginWith{\InputString}{Question}{%
		Question (\ref{#1})}{%
	\IfBeginWith{\InputString}{Footnote}{%
		Footnote \ref{#1}}{%
		\ref{#1}}}}}}}}}}}}}%
}

\definecolor{TodoRed}{RGB}{150,50,0}

\newcommand{\TextForAll}{\hspace{2pt} \text{ for all } \hspace{2pt}}
\newcommand{\TextSuchThat}{\hspace{2pt}\text{ such that }\hspace{2pt}}
\newcommand{\TextIf}{\hspace{2pt}\text{ if }\hspace{2pt}}

\newcommand{\TextAnd}{\hspace{2pt}\text{ and }\hspace{2pt}}
\newcommand{\TextElse}{\hspace{2pt}\text{ else }\hspace{2pt}}

\newcommand{\abs}[1]{\left|#1\right|}

\newcommand{\norm}[1]{\left\|#1\right\|}
\newcommand{\scprod}[2]{\langle#1,#2\rangle}

\newcommand{\ProjToIndex}[2]{\left.#2\right|_{#1}}

\newcommand{\ProjToSet}[2]{\mathcal{P}_{#1}\left(#2\right)}

\newcommand{\SetOf}[1]{\left[#1\right]}
\newcommand{\SetSize}[1]{\#\left(#1\right)}

\newcommand{\argmax}[1]{\underset{#1}{\textnormal{argmax}}}
\newcommand{\argmin}[1]{\underset{#1}{\textnormal{argmin}}}

\newcommand{\Expect}[1]{\mathbb{E}\left[#1\right]}

\def\AddBasicFunction#1#2{
	\expandafter\def\csname #1\endcsname##1{
		\def\InputString{##1}
		\def\CheckString{}
		\ifx\InputString\CheckString 
			#2
		\else
			#2 \left(##1\right)
		\fi
	}
}
\AddBasicFunction{Exp}{\exp}
\AddBasicFunction{Cos}{\cos}
\AddBasicFunction{Sin}{\sin}
\AddBasicFunction{Tan}{\tan}
\AddBasicFunction{Ln}{\ln}
\AddBasicFunction{Log}{\log}
\AddBasicFunction{ArcCos}{\arccos}
\AddBasicFunction{ArcSin}{\arcsin}
\AddBasicFunction{ArcTan}{\arctan}
\AddBasicFunction{Cosh}{\cosh}
\AddBasicFunction{Sinh}{\sinh}
\AddBasicFunction{Tanh}{\tanh}
\AddBasicFunction{ArcCosh}{\textnormal{arccosh}}
\AddBasicFunction{ArcSinh}{\textnormal{arcsinh}}
\AddBasicFunction{ArcTanh}{\textnormal{arctanh}}

\newcommand{\LogBase}[2]{\log_{#1}\left(#2\right)}

\AddBasicFunction{Kernel}{\ker}
\AddBasicFunction{Rank}{\textnormal{rank}}
\AddBasicFunction{Range}{\textnormal{Ran}}

\AddBasicFunction{supp}{\textnormal{supp}}

\AddBasicFunction{sgn}{\textnormal{sgn}}

\newcommand{\ExpE}{\mathrm{e}}

\DeclareFontFamily{U}{mathx}{\hyphenchar\font45}
\DeclareFontShape{U}{mathx}{m}{n}{
      <5> <6> <7> <8> <9> <10>
      <10.95> <12> <14.4> <17.28> <20.74> <24.88>
      mathx10
      }{}
\DeclareSymbolFont{mathx}{U}{mathx}{m}{n}
\DeclareMathSymbol{\bigtimes}{1}{mathx}{"91}
\newcommand{\normRHS}[1]{\norm{#1}}

\AddBasicFunction{Median}{\textnormal{median}}
\AddBasicFunction{RightVertices}{\textnormal{Row}}
\AddBasicFunction{LeftVertices}{\textnormal{Col}}
\AddBasicFunction{fOpt}{f}

\newcommand{\boundEmpty}{C}
\newcommand{\bound}[2]{\boundEmpty\left(#1,#2\right)}

\newcommand{\normDual}[1]{\norm{#1}_{\ast}}

\AddBasicFunction{Q}{Q_A}

\newcommand{\OrderOf}[1]{\mathcal{O}\left(#1\right)}
\newcommand{\Prox}[2]{\textnormal{Prox}_{#1}\left(#2\right)}
\AddBasicFunction{F}{F}
\AddBasicFunction{G}{G}
\AddBasicFunction{Fstar}{F^\ast}
\AddBasicFunction{Gstar}{G^\ast}
\AddBasicFunction{f}{f}
\newcommand{\tol}{\epsilon}

\newcommand{\Mean}[1]{\text{Mean}\left(#1\right)}
\newcommand{\MeanN}{\Mean{\text{N$\ell_1$E}}}
\newcommand{\MeanLN}[1]{\Mean{\text{LN$\ell_{#1}$E}}}
\newcommand{\MeanLNR}[1]{\Mean{\text{LN$\ell_{#1}$R}}}
\newcommand{\MeanNperNP}{\Mean{\text{N$\ell_1$E/$\ell_1$NP}}}
\newtheorem{CounterTheorem}{}[section]

\newtheorem{Definition}[CounterTheorem]{Definition}
\newtheorem{Theorem}[CounterTheorem]{Theorem}
\newtheorem{Proposition}[CounterTheorem]{Proposition}
\newtheorem{Lemma}[CounterTheorem]{Lemma}
\newtheorem{Corollary}[CounterTheorem]{Corollary}
\newtheorem{Remark}[CounterTheorem]{Remark}

\newtheorem{Experiment}{Experiment}

\newtheorem{AlgorithmEnvirorementForNTheorem}[CounterTheorem]{Algorithm}
\newenvironment{Algorithm}[2][]{\begin{AlgorithmEnvirorementForNTheorem}[#1]\label{#2}\hspace{1pt}\\\begin{algorithm}[H]}{\end{algorithm}\end{AlgorithmEnvirorementForNTheorem}}

\renewcommand{\Vec}[1]{\mathbf{#1}}

\newcommand{\eVec}{\Vec{e}}

\newcommand{\tVec}{\Vec{t}}
\newcommand{\uVec}{\Vec{u}}
\newcommand{\vVec}{\Vec{v}}
\newcommand{\wVec}{\Vec{w}}
\newcommand{\xVec}{\Vec{x}}
\newcommand{\yVec}{\Vec{y}}
\newcommand{\zVec}{\Vec{z}}

\newcommand{\Mat}[1]{\mathbf{#1}}

\newcommand{\AMat}{\Mat{A}}

\newcommand{\WMat}{\Mat{W}}

\newcommand{\IDVec}{\mathbbm{1}}

\title{Efficient Tuning-Free $\ell_1$-Regression of Nonnegative Compressible Signals}
\date{}
\author{
	Hendrik Bernd Petersen
	\footnote{
		Communications and Information Theory Group,
		Technische Universtität Berlin, Berlin,
		\href{mailto:petersen@tu-berlin.de}{petersen@tu-berlin.de}
	}
	\and
	Bubacarr Bah
	\footnote{
		African Institute for Mathematical Sciences (AIMS) South Africa,
		Cape Town,
		and Division of Applied Mathematics,
		Stellenbosch University, Stellenbosch
		\href{mailto:bubacarr@aims.ac.zae}{bubacarr@aims.ac.za}
	}
	\and
	Peter  Jung
	\footnote{
		Communications and Information Theory Group,
		Technische Universtität Berlin, Berlin,
		\href{mailto:peter.jung@tu-berlin.de}{peter.jung@tu-berlin.de}
	}
}

\begin{document}
	\maketitle
\begin{abstract}
In compressed sensing the goal is to recover a signal from as few as possible noisy, linear measurements.
The general assumption is that the signal has only a few non-zero entries.
The recovery can be performed by multiple different decoders, however most of them rely on some tuning.
Given an estimate for the noise level a common convex approach to recover the signal is basis pursuit denoising.
If the measurement matrix has the robust null space property with respect to the $\ell_2$-norm, basis pursuit denoising
obeys stable and robust recovery guarantees.
In the case of unknown noise levels, nonnegative least squares
recovers non-negative signals if the measurement matrix fulfills an additional property (sometimes called the $M^+$-criterion).
However, if the measurement matrix is the biadjacency matrix of a random left regular bipartite graph
it obeys with a high probability the null space property with respect to the $\ell_1$-norm with optimal parameters.
Therefore, we discuss non-negative least absolute deviation (NNLAD).
For these measurement matrices, we prove a uniform, stable and robust recovery guarantee without the need for tuning.
Such guarantees are important, since binary expander matrices are sparse and thus allow for fast sketching and recovery.
We will further present a method to solve the NNLAD numerically and show that this is comparable to state of the art methods.
Lastly, we explain how the NNLAD can be used for viral detection in the recent COVID-19 crisis.
\end{abstract}
\section{Introduction}
Since it has been realized that many signals admit a sparse representation in some frames,
the question arose whether or not such signals can be recovered from less samples than the dimension of the domain
by utilizing the low dimensional structure of the signal.
The question was already answered positively in the beginning of the millennium \cite{cs_candes}\cite{cs_donoho}.
By now there are multiple different decoders to recover a sparse signal from noisy measurements
with robust recovery guarantees.
Most of them however rely on some form of tuning, depending on either the signal or the noise.\\
The basis pursuit denoising requires an upper bound on the norm of the noise \cite[Theorem~4.22]{Introduction_CS},
the least shrinkage and selection operator
an estimate on the $\ell_1$-norm of the signal \cite[Theorem~11.1]{book_LASSO}
and the Lagrangian version of least shrinkage and selection operator
allegedly needs to be tuned
to the order of the the noise level \cite[Theorem~11.1]{book_LASSO}.
The expander iterative hard thresholding needs the sparsity of the signal
or an estimate of the order of the expansion property \cite[Theorem~13.15]{Introduction_CS}.
The order of the expansion property
can be calculated from the measurement matrix, however there is no polynomial time method known
to do this.
Variants of these methods have similar drawbacks.
The non-negative basis pursuit denoising requires the same tuning parameter as the basis pursuit denoising
\cite{non-neg_BP}.
Other thresholding based decoders like sparse matching pursuit and expander matching pursuit
have the same limitations as the expander iterative hard thresholding \cite{combinatorial_l1NSP_gilbert}.\\
If these side information is not known a priori, many decoders yield either no recovery guarantees
or, in their imperfect tuned versions, yield sub-optimal estimation errors \cite[Theorem~11.12]{Introduction_CS}.
Even though the problem of sparse recovery from under-sampled measurements has been answered long ago,
finding tuning free decoders that achieve robust recovery guarantees is still a topic of interest.\\
The most prominent achievement for that is the non-negative least squares (NNLS)
\cite{Bruckstein2008a}\cite{Donoho2010}\cite{Wang2011} \cite{Slawski2011}\cite{Slawski2013a}.
It is completely tuning free \cite{low_rank} and in \cite{NNLS_first}\cite{NNLS}
it was proven that it achieves robust recovery guarantees
if the measurement matrix consists of certain independent sub-Gaussian random variables.
\subsection{Our Contribution}
We will replace the least squares in the NNLS with an arbitrary norm
and obtain the non-negative least residual (NNLR).
By adapting \cite{NNLS_first} we prove a
recovery guarantees under similar conditions as the NNLS.
In particular, we consider the case where we minimize the $\ell_1$-norm of the residual (NNLAD)
and give a recovery guarantee if the measurement matrix is a random walk matrix of a uniformly at random drawn $D$-left regular bipartite graph.\\
In general, our result states that if a certain measurement is present,
the basis pursuit denoising can be replaced by the tuning-less NNLR
for non-negative signals.
While sub-Gaussian measurement matrices rely on
a probabilistic argument to verify that this measurement is present,
random walk matrices of left regular graphs naturally
have the measurement.
The tuning-less nature gives the NNLR an advantage over other decoders if the
noise power can not be estimated, which is for instance the case if the noise
components are multiplicative,
i.e. a random variable times the true measurements, or when the noise
is Laplacian distributed.
The latter noise distribution or the existence of outliers
also favors an $\ell_1$ regression approach over an $\ell_2$ regression approach
and thus motivate to use the NNLAD over the NNLS.\\
Further, the sparse structure of left regular graphs
can reduce the encoding and decoding time to a fraction.
Using \cite{APP} we can solve the NNLAD with a first order method of a single optimization problem with
a sparse measurement matrix. Other state of the art decoders often use non-convex optimization,
computationally complex projections or need to solve multiple different optimization problems.
For instance, to solve the basis pursuit denoising given a tuning parameter a common approach is to solve a sequence of
LASSO problems to approximate where the Pareto curve attains
the value of the tuning parameter of basis pursuit denoising \cite{SPGL1}.\\
\subsection{Relations to Other Works}
We build on the theory of \cite{NNLS_first} that uses
the $\ell_2$ null space property and the $M^+$ criterion.
These methods have also been used in \cite{low_rank}\cite{NNLS}.
To the best of the authors knowledge the
$M^+$ criterion has not been used with an $\ell_1$ null space property before.
Other works have used adjacency matrices of graphs
as measurements matrices including
\cite{combinatorial_Jafarpour}\cite{combinatorial_Xu}\cite{l1NSP_berinde}\cite{combinatorial_l1NSP_gilbert}\cite{combinatoril_l1NSP_khaje}.
The works \cite{combinatorial_Jafarpour}\cite{combinatorial_Xu}
did not consider noisy observations.
The decoder in \cite{l1NSP_berinde}
is the basis pursuit denoising and thus requires tuning depending on the
noise power. \cite{combinatoril_l1NSP_khaje} proposes two decoders for
non-negative signals.
The first is the non-negative basis pursuit which could be extended
to the non-negative basis pursuit denoising. However this again
needs a tuning parameter depending on the noise power.
The second decoder, the Reverse Expansion Recovery algorithm,
requires the order of the expansion property,
which is not known to be calculatable in a polynomial time.
The survey \cite{combinatorial_l1NSP_gilbert}
contains multiple decoders including the basis pursuit,
which again needs tuning depending on the noise power for robustness,
the expander matching pursuit and the sparse matching pursuit,
which need the order of the expansion property.
Further, \cite{non-neg_BP} considered sparse regression of non-negative signals
and also used the non-negative basis pursuit denoising as decoder,
which again needs tuning dependent on the noise power.
To the best of the authors knowledge, this is the first work
that considers tuning-less sparse recovery for random walk matrices of left
regular bipartite graphs.
The NNLAD has been considered in \cite{super_res} with a structured sparsity model
without the use of the $M^+$ criterion.
\section{Preliminaries}
For $K\in\mathbb{N}$ we denote the set of integers from $1$ to $K$ by $\SetOf{K}$.
For a set $T\subset\SetOf{N}$ we denote the number of elements in $T$ by $\SetSize{T}$.
Vectors are denoted by lower case bold face symbols,
while its corresponding components are denoted by lower case italic letters.
Matrices are denoted by upper case bold face symbols,
while its corresponding components are denoted by upper case italic letters.
For $\xVec\in\mathbb{R}^N$ we denote its $\ell_p$-norms by $\norm{\xVec}_p$.
Given $\AMat\in\mathbb{R}^{M\times N}$ we denote its operator norm as
operator from $\ell_q$ to $\ell_p$ by
$\norm{\AMat}_{q\rightarrow p}:=\sup_{\vVec\in\mathbb{R}^N,\norm{\vVec}_q\leq 1}\norm{\AMat\vVec}_p$.
By $\mathbb{R}_+^N$ we denote the non-negative orthant.
Given a closed convex set $C\subset\mathbb{R}^N$, we denote the projection onto $C$,
i.e. the unique minimizer of $\argmin{\zVec\in C}\frac{1}{2}\norm{\zVec-\vVec}_2^2$, by 
$\ProjToSet{C}{\vVec}$.
For a vector $\xVec\in\mathbb{R}^N$ and a set $T\subset\SetOf{N}$, $\ProjToIndex{T}{\xVec}$
denotes the vector in $\mathbb{R}^N$, whose $n$-th component is $x_n$ if $n\in T$ and $0$ else.
Given $N,S\in\mathbb{N}$ we will often need sets $T\subset\SetOf{N}$ with $\SetSize{T}\leq S$
and we abbreviate this by $\SetSize{T}\leq S$ if no confusion is possible.\\
Given a measurement matrix $\AMat\in\mathbb{R}^{M\times N}$
a decoder is any map $\Q{}:\mathbb{R}^M\rightarrow\mathbb{R}^N$.
A signal is any possible $\xVec\in\mathbb{R}^N$.
If $\xVec\in\mathbb{R}^N_+=\left\{\zVec\in\mathbb{R}^N:z_n\geq 0 \TextForAll n\in\SetOf{N}\right\}$,
we say the signal is non-negative and write shortly $\xVec\geq 0$.
If additionally $x_n>0$ for all $n\in\SetOf{N}$, we write $\xVec>0$.
An observation is any possible input of a decoder, i.e. all $\yVec\in\mathbb{R}^M$.
We allow all possible inputs of the decoder as observation, since in general the transmitted
codeword $\AMat\xVec$ is disturbed by some noise.
Thus, given a signal $\xVec$ and an observation $\yVec$ we call
$\eVec:=\yVec-\AMat\xVec$ the noise.
A signal $\xVec$ is called $S$-sparse if $\norm{\xVec}_0:=\SetSize{\left\{n\in\SetOf{N}:x_n\neq 0\right\}}\leq S$.
We denote the set of $S$-sparse vectors by
\begin{align}\notag
	\Sigma_S:=\left\{\zVec\in\mathbb{R}^N:\norm{\zVec}_0\leq S\right\}.
\end{align}
Given some $S\in\SetOf{N}$ the compressibility of a signal $\xVec$
can be measured by $d_1\left(\xVec,\Sigma_S\right):=\inf_{\zVec\in\Sigma_S}\norm{\xVec-\zVec}_1$.\\
Given $N$ and $S$, the general non-negative compressed sensing task is to find a measurement matrix
$\AMat\in\mathbb{R}^{M\times N}$ and a decoder $\Q{}:\mathbb{R}^{M}\rightarrow\mathbb{R}^{N}$
with $M$ as small as possible such that the following holds true:
There exists a $q\in\left[1,\infty\right]$ and a continuous function $\boundEmpty:\mathbb{R}\times\mathbb{R}^M\rightarrow\mathbb{R}_+$
with $\bound{0}{0}=0$ such that
\begin{align}\notag
	\norm{\Q{\yVec}-\xVec}_q\leq \bound{d_1\left(\xVec,\Sigma_S\right)}{\yVec-\AMat\xVec}
	\TextForAll\xVec\in\mathbb{R}_+^N \TextAnd \yVec\in\mathbb{R}^M
\end{align}
holds true.
This will ensure that if we can control the compressibility and the noise,
we can also control the estimation error and in particular decode every noiseless observation of $S$-sparse signals exactly.
\section{Main Results}\label{Section:main_results}
Given a measurement matrix $\AMat\in\mathbb{R}^{M\times N}$ and a norm $\normRHS{\cdot}$ on $\mathbb{R}^M$
we propose to define the decoder as follows:
Given $\yVec\in\mathbb{R}^M$ set $\Q{\yVec}$ as any minimizer of
\begin{align}\label{Problem:NNLN}\tag{NNLR}
	\argmin{\zVec\geq 0}\normRHS{\AMat\zVec-\yVec}.
\end{align}
We call this problem non-negative least residual (NNLR).
In particular, for $\normRHS{\cdot}=\norm{\cdot}_1$ this problem is called
non-negative least absolute deviation (NNLAD) and for $\normRHS{\cdot}=\norm{\cdot}_2$ 
this problem is known as the non-negative least squares (NNLS) studied in \cite{NNLS_first}.
In fact, we can translate the proof techniques fairly simple.
We just need to introduce the dual norm.
\begin{Definition}
	Let $\normRHS{\cdot}$ be a norm on $\mathbb{R}^M$. The norm $\normDual{\cdot}$ on $\mathbb{R}^M$ defined by
	$
		\normDual{\vVec}:=\sup_{\normRHS{\uVec}\leq 1}\scprod{\vVec}{\uVec},
	$
	is called dual norm to $\normRHS{\cdot}$.
\end{Definition}
Note that the dual norm is actually a norm.
To obtain a recovery guarantee for NNLR we have certain requirements on the measurement matrix $\AMat$.
As for most other convex optimization problems in compressed sensing, we use a null space property.
\begin{Definition}
	Let $S\in\SetOf{N}$, $q\in\left[1,\infty\right)$ and $\normRHS{\cdot}$ be any norm on $\mathbb{R}^M$.
	Further let $\AMat\in\mathbb{R}^{M\times N}$.
	Suppose there exists constants $\rho\in\left[0,1\right)$ and $\tau\in\left[0,\infty\right)$ such that
	\begin{align}\notag
		\norm{\ProjToIndex{T}{\vVec}}_q
		\leq \rho S^{\frac{1}{q}-1}\norm{\ProjToIndex{T^c}{\vVec}}_1 + \tau\normRHS{\AMat \vVec}
		\TextForAll \vVec\in\mathbb{R}^N \TextAnd \SetSize{T}\leq S.
	\end{align}
	Then, we say $\AMat$ has the $\ell_q$-robust null space property of order $S$ with respect to
	$\normRHS{\cdot}$ or in short $\AMat$ has the $\ell_q$-RNSP of order $S$ with respect to $\normRHS{\cdot}$
	with constants $\rho$ and $\tau$.
	$\rho$ is called stableness constant and $\tau$ is called robustness constant.
\end{Definition}
In order to deal with the non-negativity, we need $\AMat$ to be biased in a certain way.
In \cite{NNLS_first} this bias was guaranteed with the $M^+$ criterion.
\begin{Definition}
	Let $\AMat\in\mathbb{R}^{M\times N}$.
	Suppose there exists $\tVec\in\mathbb{R}^M$ such that $\AMat^T\tVec>0$.
	Then we say $\AMat$ obeys the the $M^+$ criterion with vector $\tVec$ and constant
	$\kappa:=\max_{n\in\SetOf{N}}\abs{\left(\AMat^T\tVec\right)_n}\max_{n\in\SetOf{N}}\abs{\left(\AMat^T\tVec\right)^{-1}_n}$.
\end{Definition}
Note that $\kappa$ is actually a condition number of the matrix with diagonal $\AMat^T\tVec$
and $0$ else. Condition number numbers are frequently used in error bounds of numerical linear algebra.
The general recovery guarantee is the following and
similar results have been obtained in the matrix case in \cite{NNLR_matrix}.
\begin{Theorem}[NNLR Recovery Guarantee]\label{Theorem:NNLDMinimizer}
	Let $S\in\SetOf{N}$, $q\in\left[1,\infty\right)$ and $\normRHS{\cdot}$ be any norm on $\mathbb{R}^M$ with dual norm $\normDual{\cdot}$.
	Further, suppose that $\AMat\in\mathbb{R}^{M\times N}$ obeys
	\begin{itemize}
		\item[a)]
			the $\ell_q$-RNSP of order $S$ with respect to $\normRHS{\cdot}$ with constants $\rho$ and $\tau$ and
		\item[b)] 
			the $M^+$ criterion with vector $\tVec$ and constant $\kappa$.
	\end{itemize}
	If $\kappa\rho<1$, the following recovery guarantee holds true:
	For all $\xVec\in\mathbb{R}_+^N$ and $\yVec\in\mathbb{R}^M$
	any minimizer $\xVec^\#$ of
	\begin{align}\tag{NNLR}
		\argmin{\zVec\geq 0}\normRHS{\AMat\zVec-\yVec}
	\end{align}
	obeys the bound
	\begin{align}\notag
		\norm{\xVec-\xVec^\#}_q
		\leq&2\frac{\left(1+\kappa\rho\right)^2}{1-\kappa\rho}
			\kappa S^{\frac{1}{q}-1}
			d_1\left(\xVec,\Sigma_S\right)
		+2\left(
				\frac{\left(1+\kappa\rho\right)^2}{1-\kappa\rho} S^{\frac{1}{q}-1}\max_{n\in\SetOf{N}}\abs{\left(\AMat^T\tVec\right)^{-1}_n}
				\normDual{\tVec}+\frac{3+\kappa\rho}{1-\kappa\rho}\kappa\tau
			\right)\normRHS{\AMat\xVec-\yVec}.
	\end{align}
	If $q=1$, this bound can be improved to
	\begin{align}\notag
		\norm{\xVec-\xVec^\#}_1
		\leq&2\frac{1+\kappa\rho}{1-\kappa\rho}
			\kappa
			d_1\left(\xVec,\Sigma_S\right)
		+2\left(
				\frac{1+\kappa\rho}{1-\kappa\rho}\max_{n\in\SetOf{N}}\abs{\left(\AMat^T\tVec\right)^{-1}_n}
				\normDual{\tVec}+\frac{2}{1-\kappa\rho}\kappa\tau
			\right)\normRHS{\AMat\xVec-\yVec}.
	\end{align}
\end{Theorem}
\begin{proof}
	The proof can be found in \refP{Subsection:Proof:NNLDMinimizer}.
\end{proof}
Given a matrix with $\ell_q$-RNSP we can
add a row of ones (or a row consisting of one minus the column sums of the matrix)
to fulfill the $M^+$ criterion with the optimal $\kappa=1$.
Certain random measurement matrices guarantee uniform bounds on $\kappa$ for fixed vectors $\tVec$.
In \cite[Theorem~12]{NNLS_first} it was proven that if $A_{m,n}$ are all i.i.d. $0/1$ Bernoulli random variables,
$\AMat$ has $M^+$ criterion with $\tVec=\left(1,\dots,1\right)^T\in\mathbb{R}^M$ and $\kappa\leq 3$
with high probability.
This is problematic, since if $\kappa>1$, it might happen that $\kappa\rho<1$ is not fulfilled anymore.
Since the stableness constant $\rho\left(S'\right)$ as a function of $S'$ is monotonically increasing,
the condition $\kappa\rho(S')<1$ might only hold if $S'<S$. If that is the case,
there are vectors $\xVec\in\Sigma_S$ that are being recovered by basis pursuit denoising but not by NNLS!
This is for instance the case for the matrix
$\AMat=\begin{pmatrix}
	1 & 0 & 1 \\
 	0 & 1 & 1
\end{pmatrix}$, which has $\ell_1$-robust null space property of order $1$ with stableness constant
$\rho:=\frac{1}{2}$ and $M^+$ criterion with $\kappa\geq 2$ for any possible choice of $\tVec$.
In particular, the vector $\xVec=\left(0,0,1\right)^T$ is not necessarily
being recovered by the NNLAD and the NNLS.\\
Hence, it is crucial that the vector $\tVec$ is chosen to minimize $\kappa$
and ideally obeys the optimal $\kappa=1$.
This motivates us to use random walk matrices of regular graphs
since they obey exactly this.
\begin{Definition}\label{Definition:LosslessExpander}
	Let $\AMat\in\left\{0,1\right\}^{M\times N}$ and $D\in\SetOf{M}$.
	For $T\subset N$ the set
	\begin{align}\notag
		\RightVertices{T}
		:=\bigcup_{n\in T}\left\{m\in\SetOf{M}\TextSuchThat A_{m,n}=1\right\}
	\end{align}
	is called the set of right vertices connected to the set of left vertices $T$.
	If
	\begin{align}\notag
		\SetSize{\RightVertices{ \left\{n\right\} }}
		=D
		\TextForAll n\in\SetOf{N},
	\end{align}
	then $D^{-1}\AMat\in\left\{0,D^{-1}\right\}^{M\times N}$
	is called a random walk matrix of a $D$-left regular bipartite graph.
	We also say short that $D^{-1}\AMat$ is a $D$-LRBG.
	If additionally there exists a $\theta\in\left[0,1\right)$ such that
	\begin{align}\notag
		\SetSize{\RightVertices{T}}
		\geq\left(1-\theta\right)D \SetSize{T}
		\TextForAll \SetSize{T}\leq S
	\end{align}
	holds true, then $D^{-1}\AMat$ is called a random walk matrix of a
	$\left(S,D,\theta\right)$-lossless expander.
\end{Definition}
We will only consider random walk matrices and no biadjacency matrices.
Note that we have made a slight abuse of notation. The term $D$-LRBG as a short form for $D$-left regular bipartite graph
refers in our case to the random walk matrix $\AMat$ but not the graph itself. We omit this minor technical differentiation,
for the sake of shortening the frequently used term random walk matrix of a $D$-left regular bipartite graph.
Lossless expanders are bipartite graphs that have a low number of edges but are still highly connected,
see for instance \cite[Chapter~4]{pseudorandomness}.
As a consequence their random walk matrices have good properties for compressed sensing.
It is well known that random walk matrices of a $\left(2S,D,\theta\right)$-lossless expanders obey the $\ell_1$-RNSP
of order $S$ with respect to $\norm{\cdot}_1$, see \cite[Theorem~13.11]{Introduction_CS}.
The dual norm of $\norm{\cdot}_1$ is the norm $\norm{\cdot}_\infty$
and the $M^+$ criterion is easily fulfilled, since the columns sum up to one.
From \thref{Theorem:NNLDMinimizer} we can thus draw the following corollary.
\begin{Corollary}[Lossless Expander Recovery Guarantee]\label{Corollary:lossless_expander}
	Let $S\in\SetOf{N}$, $\theta\in\left[0,\frac{1}{6}\right)$.
	If $\AMat\in\left\{0,D^{-1}\right\}^{M\times N}$ is a random walk matrix of a $\left(2S,D,\theta\right)$-lossless expander,
	then the following recovery guarantee holds true:
	For all $\xVec\in\mathbb{R}_+^N$ and $\yVec\in\mathbb{R}^M$
	any minimizer $\xVec^\#$ of
	\begin{align}\tag{NNLAD}
		\argmin{\zVec\geq 0}\norm{\AMat\zVec-\yVec}_1
	\end{align}
	obeys the bound
	\begin{align}\label{Equation:EQ1:Corollary:lossless_expander}
		\norm{\xVec-\xVec^\#}_1
		\leq&2\frac{1-2\theta}{1-6\theta}
			d_1\left(\xVec,\Sigma_S\right)
		+2\frac{3-2\theta}{1-6\theta}\norm{\AMat\xVec-\yVec}_1.
	\end{align}
\end{Corollary}
\begin{proof}
	By \cite[Theorem~13.11]{Introduction_CS} $\AMat$ has $\ell_1$-RNSP with respect to $\norm{\cdot}_1$
	with constants $\rho=\frac{2\theta}{1-4\theta}$ and $\tau=\frac{1}{1-4\theta}$.
	The dual norm of the norm $\norm{\cdot}_1$ is $\norm{\cdot}_\infty$.
	If we set $\tVec:=\left(1,\dots,1\right)^T\in\mathbb{R}^M$, we get
	\begin{align}\notag
		\left(\AMat^T\tVec\right)_n
		=\sum_{m\in\SetOf{M}}A_{m,n}
		=D D^{-1}
		=1 \TextForAll n\in\SetOf{N}.
	\end{align}
	Hence, $\AMat$ has the $M^+$ criterion with vector $\tVec$ and constant $\kappa=1$
	and the condition $\kappa\rho<1$ is immediately fulfilled.
	We obtain $\normDual{\tVec}=\norm{\tVec}_\infty=1$ and $\max_{n\in\SetOf{N}}\abs{\left(\AMat^T\tVec\right)^{-1}_n}=1$.
	Applying \thref{Theorem:NNLDMinimizer} with improved bound for $q=1$ and these values yields
	\begin{align}\notag
		\norm{\xVec-\xVec^\#}_1
		\leq&2\frac{1+\rho}{1-\rho}
			d_1\left(\xVec,\Sigma_S\right)
		+2\left(\frac{1+\rho}{1-\rho}+\frac{2}{1-\rho}\tau\right)\norm{\AMat\xVec-\yVec}_1.
	\end{align}
	If we additionally substitute the values for $\rho$ and $\tau$ we get
	\begin{align}\notag
		\norm{\xVec-\xVec^\#}_1
		\leq&2\frac{1-2\theta}{1-6\theta}
			d_1\left(\xVec,\Sigma_S\right)
		+2\left(\frac{1-2\theta}{1-6\theta}+2\frac{1}{1-6\theta}\right)\norm{\AMat\xVec-\yVec}_1
		\\\notag
		\leq&2\frac{1-2\theta}{1-6\theta}
			d_1\left(\xVec,\Sigma_S\right)
		+2\frac{3-2\theta}{1-6\theta}\norm{\AMat\xVec-\yVec}_1.
	\end{align}
	This finishes the proof.
\end{proof}
Note that \cite[Theorem~13.11]{Introduction_CS} is an adaption of
\cite[Lemma~11]{l1NSP_berinde} to account for robustness
and skips proving the $\ell_1$ restricted isometry property.
If $M\geq \frac{2}{\theta}\Exp{\frac{2}{\theta}}S\Ln{\frac{\ExpE N}{S}}$ and
$D=\left\lceil\frac{2}{\theta}\Ln{\frac{\ExpE N}{S}}\right\rceil$, a uniformly at random drawn $D$-LRBG
is a random walk matrix of a $\left(2S,D,\theta\right)$-lossless expander
with a high probability \cite[Theorem~13.7]{Introduction_CS}.
Thus, recovery with the NNLAD is possible in the optimal regime $M\in\OrderOf{S\Log{\frac{N}{S}}}$.
\subsection*{On the Robustness Bound for Lossless Expanders}
If $\AMat$ is a random walk matrix of a $\left(2S,D,\theta\right)$-lossless expander with $\theta\in\left[0,\frac{1}{6}\right)$,
then we can also draw a recovery guarantee for the NNLS.
By \cite[Theorem~13.11]{Introduction_CS} $\AMat$ has $\ell_1$-RNSP with respect
to $\norm{\cdot}_1$	with constants $\rho=\frac{2\theta}{1-4\theta}$ and $\tau=\frac{1}{1-4\theta}$
and hence also $\ell_1$-RNSP with respect to $\norm{\cdot}_2$ with constants $\rho'=\rho$ and $\tau'=\tau M^\frac{1}{2}$.
Similar to the proof of \thref{Corollary:lossless_expander} we can use \thref{Theorem:NNLDMinimizer} to deduce that
any minimizer $\xVec^\#$ of
\begin{align}\tag{NNLS}
	\argmin{\zVec\geq 0}\norm{\AMat\zVec-\yVec}_2,
\end{align}
obeys the bound
\begin{align}\label{Equation:EQ2:Corollary:lossless_expander}
	\norm{\xVec-\xVec^\#}_1
	\leq&2\frac{1-2\theta}{1-6\theta}
		d_1\left(\xVec,\Sigma_S\right)
	+2\frac{3-2\theta}{1-6\theta}M^\frac{1}{2}\norm{\AMat\xVec-\yVec}_2.
\end{align}
If the measurement error $\eVec=\yVec-\AMat\xVec$ is a constant vector,
i.e. $\eVec=\alpha\IDVec$, then $\norm{\eVec}_1=M^\frac{1}{2}\norm{\eVec}_2$. In this case
the error bound of the NNLS is just as good as the error bound of the NNLAD.
However, if $\eVec$ is a standard unit vector, then $\norm{\eVec}_1=\norm{\eVec}_2$.
In this case the error bound of the NNLS is significantly worse than the error bound of the NNLAD.
Thus, the NNLAD performs better under peaky noise, while the NNLS and NNLAD are tied under noise with evenly
distributed mass. We will verify this numerically in  \refP{Subsection:Properties_of_NNLAD}.
One can draw a complementary result for matrices with biased sub-Gaussian entries, which obey
the $\ell_2$-RNSP with respect to $\norm{\cdot}_2$ and the $M^+$ criterion in the optimal regime \cite{NNLS_first}.
\refP{Table:main_results:table_of_advantages} states the methods, which have an advantage over the other in each scenario.
\begin{table}[ht]
	\centering
	\begin{tabular}{c|c|c|c}
		\multicolumn{1}{c}{}
		&
		& \multicolumn{2}{c}{Measurement Matrix}
		
		\\\cline{3-4}
		\multicolumn{1}{c}{}
		&
		& $D$-LRBG ($\ell_1$)
		& biased sub-Gaussian ($\ell_2$)
		\\\cline{1-4}
		
		& peaky $\norm{\eVec}_1\approx\norm{\eVec}_2$
		& NNLAD
		& -
		\\\cline{2-4}
		Noise
		& even mass $\norm{\eVec}_1\approx M^\frac{1}{2}\norm{\eVec}_2$
		& -
		& NNLS
		\\\cline{2-4}
		
		& unknown noise
		& NNLAD
		& NNLS
	\end{tabular}
	\caption{\label{Table:main_results:table_of_advantages}
		Table of advantages of NNLAD and NNLS over each other.}
\end{table}
\section{NNLAD using a Proximal Point Method}\label{Section:APP}
In this section we assume that $\normRHS{\cdot}=\norm{\cdot}_p$ with some
$p\in\left[1,\infty\right]$.
If $p\in\left\{1,\infty\right\}$, the NNLR can be recast as a linear program by introducing some slack variables.
For an arbitrary $p$ the NNLR is a convex optimization problem and the objective function has a
simple and globally bounded subdifferential.
Thus, the NNLR can directly be solved with a projective subgradient method using a problem independent step size.
Such subgradient methods achieve only a convergence rate of $\OrderOf{\Log{k}k^{-\frac{1}{2}}}$ towards the optimal objective value
\cite[Section~3.2.3]{nesterov_book},
where $k$ is the number of iterations performed.
In the case that the norm is the $\ell_2$-norm, we can transfer the problem into a differentiable version, i.e. the NNLS
\begin{align}\notag
	\argmin{\zVec\geq 0} \frac{1}{2}\norm{\AMat\zVec-\yVec}_2^2.
\end{align}
Since the gradient of such an objective is Lipschitz, this problem can be solved by a projected gradient method with constant step size,
which achieves a convergence rate of $\OrderOf{k^{-2}}$ towards the optimal objective value
\cite{NNLS_fista}\cite{NNLS_first_order}.
However this does not generalize to the $\ell_1$-norm.
The proximal point method proposed in \cite{APP}
can solve the case of the $\ell_1$-norm with a convergence rate $\OrderOf{k^{-1}}$
towards the optimal objective value.
This results in the following algorithm.
\begin{Algorithm}[NNLAD as First Order Method]{Algorithm:NNLAD_APP_no_average}
	\KwData{measurement $\yVec\in\mathbb{R}^M$, measurement matrix $\AMat\in\mathbb{R}^{M\times N}$,
		parameters $\sigma>0$, $\tau>0$, initializations $\xVec^0\in\mathbb{R}^N$, $\wVec^0\in\mathbb{R}^M$,
		tolerance parameters $\tol_1\geq 0$,$\tol_2\geq 0$}
	\KwResult{estimator $\xVec^\#\in\mathbb{R}^N$}
	initialize iterates\;
	$\xVec\leftarrow\xVec^0$;
	$\vVec\leftarrow\xVec^0$;
	$\wVec\leftarrow\wVec^0$\;
	initialize images\;
	$\tilde{\wVec}\leftarrow\AMat^T\wVec$;
	$\tilde{\xVec}\leftarrow\AMat\xVec$;
	$\tilde{\vVec}\leftarrow\AMat\vVec$\;
	\While{
		$\norm{\tilde{\xVec}-\yVec}_1+\scprod{\yVec}{\wVec}>\tol_1$
		\KwSty{or} $\min_{n\in\SetOf{N}}\tilde{\wVec}_n<-\tol_2$
	}{
		calculate iterates\;
		$\wVec\leftarrow\wVec+\sigma\left(\tilde{\vVec}-\yVec\right)$\;
		$\wVec\leftarrow\left(\min\left\{1,\abs{w_m}\right\}
			\sgn{w_m}\right)_{m\in\SetOf{M}}$\;
		$\tilde{\wVec}\leftarrow\AMat^T\wVec$\;
		$\vVec\leftarrow-\xVec$\;
		$\xVec\leftarrow\left(\max\left\{0,x_n-\tau\tilde{w}_n\right\}\right)_{n\in\SetOf{N}}$\;
		$\vVec\leftarrow\vVec+2\xVec$\;
		$\tilde{\vVec}\leftarrow\AMat\vVec$\;
		$\tilde{\xVec}\leftarrow\frac{1}{2}\left(\tilde{\vVec}+\tilde{\xVec}\right)$\;
	}
	\KwRet{$\xVec^\#\leftarrow\xVec$}
\end{Algorithm}
The following convergence guarantee can be deduced from \cite[Theorem~1]{APP}.
Let $\sigma\tau<\norm{\AMat}_{2\rightarrow 2}^{-2}$ and
let $\xVec^k$ and $\wVec^k$ be the values of $\xVec$ and $\wVec$
at the end of the $k$-th iteration of the while loop
of \thref{Algorithm:NNLAD_APP_no_average}.
Then, the following statements hold true:
	\begin{itemize}
		\item[(1)]
			The iterates converge:
			The sequence $\left(\xVec^k\right)_{k\in\mathbb{N}}$
			converges to a minimizer of $\argmin{\zVec\geq 0} \norm{\AMat\zVec-\yVec}_1$.
		\item[(2)]
			The iterates are feasible:
			We have $\xVec^k\geq 0$ and $\norm{\wVec^k}_\infty\leq 1$
			for all $k\geq 1$.
		\item[(3)]
			There is a stopping criteria for the iterates:\\
			$\lim_{k\rightarrow\infty}\norm{\AMat\xVec^k -\yVec}_1+\scprod{\yVec}{\wVec^k}=0$
			and $\lim_{k\rightarrow\infty}\AMat^T\wVec^k\geq 0$.
			In particular,
			if $\norm{\AMat\xVec^k -\yVec}_1+\scprod{\yVec}{\wVec^k}\leq 0$ and $\AMat^T\wVec^k\geq 0$,
			then $\xVec^k$ is a minimizer of $\argmin{\zVec\geq 0} \norm{\AMat\zVec-\yVec}_1$.
		\item[(5)]
			The averages obey the convergence rate towards the optimal objective value:\\
			$\norm{\AMat\frac{1}{k}\sum_{k'=1}^k\xVec^{k'}-\yVec}_1-\norm{\AMat\xVec^\#-\yVec}_1
			\leq \frac{1}{k}\left(\frac{1}{2\tau}\norm{\xVec^\#-\xVec^0}_2^2
				+\frac{1}{2\sigma}\left(\norm{\wVec^0}_2^2+2\norm{\wVec^0}_1+M\right)\right)$,
			where $\xVec^\#$ is a minimizer of $\argmin{\zVec\geq 0} \norm{\AMat\zVec-\yVec}_1$.
	\end{itemize}
The formal version and proof is given in the appendix.
Note that this yields a convergence guarantee for both
the iterates and averages, but the convergence rate is only guaranteed for the averages.
\thref{Algorithm:NNLAD_APP_no_average} is
optimized in the sense that it uses the least possible number of matrix vector multiplications per iteration,
since these govern the computational complexity.
\begin{Remark}
	Let $\AMat$ be $D$-LRBG.
	Each iteration of \thref{Algorithm:NNLAD_APP_no_average} requires at most
	$4DN+8N+16M$ floating point operations and $5N+4M$ assignments.
\end{Remark}
\subsection*{Iterates or Averages}
The question arises whether or not it is better to estimate with averages or iterates. Numerical testing suggest that the
iterates reach tolerance thresholds significantly faster than the averages. We can only give a heuristically explanation for this phenomenon.
The stopping criteria of the iterates yields $\lim_{k\rightarrow\infty}\AMat^T\wVec^k\geq 0$.
In practice we observe that $\AMat^T\wVec^k\geq 0$ for all sufficiently large $k$.
However, $\AMat^T\wVec^{k+1}\geq 0$ yields $\xVec^{k+1}\leq \xVec^k$.
This monotonicity promotes the converges of the iterates and gives a clue why the iterates seem to converge better in practice.
See \refP{Figure:Numcerics:computational_complexity:LRBG:1} and
\refP{Figure:Numcerics:computational_complexity:LRBG:2}.
\subsection*{On the Convergence Rate}
As stated the NNLS achieves the convergence rate $\OrderOf{k^{-2}}$ \cite{NNLS_first_order} while the
NNLAD only achieves the convergence rate of $\OrderOf{k^{-1}}$ towards to optimal objective value.
However, this should not be considered as weaker, since the objective function of the NNLS
is the square of a norm. If $\xVec^k$ are the iterates of the NNLS implementation of \cite{NNLS_first_order},
algebraic manipulation yields
\begin{align}\notag
	\norm{\AMat\xVec^k-\yVec}_2-\norm{\AMat\xVec^\#-\yVec}_2
	\leq 2^\frac{1}{2}\left(\frac{1}{2}\norm{\AMat\xVec^k-\yVec}_2^2
		-\frac{1}{2}\norm{\AMat\xVec^\#-\yVec}_2^2\right)^\frac{1}{2}
	\leq 2^\frac{1}{2}\left(Ck^{-2}\right)^\frac{1}{2}
	\leq \left(2C\right)^\frac{1}{2}k^{-1}.
\end{align}
Thus, the $\ell_2$-norm of the residual of the NNLS iterates only decays in the same
order as the $\ell_1$-norm of the residual of the NNLAD averages.
\section{Numerical Experiments and Applications}\label{Section:Numerics}
In the first part of this section we will compare NNLAD with several state of the art
recovery methods in terms of achieved sparsity levels and 
decoding time.
For $p\in\left[1,\infty\right]$, we denote $\mathbb{S}_p^{N-1}:=\left\{\zVec\in\mathbb{R}^N:\norm{\zVec}_p=1\right\}$,
and $\mathbb{S}_0^{N-1}:=\left\{\zVec\in\mathbb{R}^N:\norm{\zVec}_0=1=\norm{\zVec}_2\right\}=\Sigma_1\cap\mathbb{S}_2^{N-1}$.
\subsection{Properties of the NNLAD Optimizer}\label{Subsection:Properties_of_NNLAD}
We recall that the goal is to recover $\xVec$ from the noisy linear measurements
$\yVec=\AMat\xVec+\eVec$.
To investigate properties of the minimizers of NNLAD we compare it to the minimizers
of the well studied problems basis pursuit (BP), optimally tuned basis pursuit denoising (BPDN), optimally tuned
$\ell_1$-constrained least residual (CLR) and the NNLS, which are given by
\begin{align}\tag{BPDN}
	\label{Problem:BPDN}
	&\argmin{\zVec:\norm{\AMat\zVec-\yVec}_1\leq \epsilon}\norm{\zVec}_1 & \text{ with } \epsilon=\norm{\eVec}_1,
	\\\tag{CLR}
	\label{Problem:CLR}
	&\argmin{\zVec:\norm{\zVec}_1\leq \tau}\norm{\AMat\zVec-\yVec}_1 & \text{ with } \tau=\norm{\xVec}_1,
	\\\tag{BP}
	\label{Problem:BP}
	&\argmin{\zVec:\AMat\zVec=\yVec}\norm{\zVec}_1, &
	\\\tag{NNLS}
	\label{Problem:NNLS}
	&\argmin{\zVec\geq 0}\norm{\AMat\zVec-\yVec}_2. &
\end{align}
Further, we compare the NNLAD to any cluster point of the sequence of the expander iterative hard
thresholding (EIHT) given by
\begin{align}\tag{EIHT}
	\xVec^0:=0\TextAnd \xVec^{k+1}:=\ProjToSet{\Sigma_{S'}}{\xVec^k+\Median{\yVec-\AMat\xVec^k}}
	\TextForAll k\in\mathbb{N}_0
	\text{ and with } S'=\norm{\xVec}_0	,
\end{align}
where $\Median{\zVec}_n$ is the median of $\left(\zVec_m\right)_{m\in\RightVertices{\left\{n\right\}}}$
and $\ProjToSet{\Sigma_S}{\vVec}$ is a hard thresholding operator, i.e.
some minimizer of $\argmin{\zVec\in\Sigma_S}\frac{1}{2}\norm{\zVec-\vVec}_2^2$.
There is a whole class of thresholding based decoders for lossless expanders,
which all need either the sparsity of the signal or the order of the expansion property
as tuning parameter. We choose the EIHT as a represent of this class, since it has
robust recovery guarantees \cite[Theorem~13.5]{Introduction_CS}.
By convex decoders we refer to BPDN, BP, CLR, NNLAD, and NNLS.
We choose the optimal tuning $\epsilon=\norm{\eVec}_1$ for the BPDN and $\tau=\norm{\xVec}_1$ for the CLR.
The optimally tuned BPDN and CLR are representing a best case benchmark.
In \cite[Figure~1.1]{BPDN_misstuning} it was noticed that tuning the BPDN with $\epsilon>\norm{\eVec}_p$ often leads to
worse estimation errors than tuning with $\epsilon<\norm{\eVec}_p$ for $p=2$.
Thus, BP is a version of BPDN with no prior knowledge about the noise
and represents a worst case benchmark.
At the moment we do not care about the method to calculate the minimizers
of the optimization problems, thus
we solve all optimization problems with the CVX package of Matlab \cite{CVX1}, \cite{CVX2}.
For a given $SNR,r,N,M,D,S$ we will do the following experiment multiple times:
\begin{Experiment}\label{Experiment:test:1}
	\hspace{1pt}
	\begin{itemize}
		\item[1.] Generate a measurement matrix $\AMat\in\left\{0,D^{-1}\right\}^{M\times N}$ as a uniformly at random drawn $D$-LRBG.
		\item[2.] Generate a signal $\xVec$ uniformly at random from $\Sigma_S\cap\mathbb{R}_+^N\cap\mathbb{S}_1^{N-1}$.
		\item[3.] Generate a noise $\eVec$ uniformly at random from $\frac{\norm{\AMat\xVec}_1}{SNR}\mathbb{S}_r^{M-1}$.
		\item[4.] Define the observation $\yVec:=\AMat\xVec+\eVec$.
		\item[5.] For each decoder $\Q{}$ calculate an estimator $\xVec^\#:=\Q{\yVec}$ and collect the relative estimation error
			$\norm{\xVec-\xVec^\#}_1=\frac{\norm{\xVec-\xVec^\#}_1}{\norm{\xVec}_1}$.
	\end{itemize}
\end{Experiment}
In this experiment we have $SNR=\frac{\norm{\AMat\xVec}_1}{\norm{\eVec}_1}$ and since
$\AMat$ is a $D$-LRBG and $\xVec\geq 0$, we further have $\norm{\AMat\xVec}_1=\norm{\xVec}_1=1$. 
Note that for $r=0$ and $r=1$ we obtain two different noise distributions.
If $\eVec$ is uniformly distributed on $\mathbb{S}_1^{M-1}$, then
the absolute value of each component $\abs{e_m}$ is a random variable with density
$h\mapsto\left(M-1\right)\left(1-h\right)^{M-2}$ for $h\in\left[0,1\right]$.
Thus, $\Expect{\norm{\eVec}_2^2}=M\frac{2}{M\left(M+1\right)}=\frac{2}{M+1}$.
By testing one can observe a concentration around this expected value,
in particular that $M^\frac{1}{2}\norm{\eVec}_2\approx \sqrt{2}\norm{\eVec}_1$
with a high probability.
If $\eVec$ is uniformly distributed on $\mathbb{S}_0^{M-1}$, then
$\norm{\eVec}_2=\norm{\eVec}_1$.
Thus, these two noise distributions each represent randomly drawn noise vectors
obeying one norm equivalence asymptotically tightly up to a constant.
From \refP{Equation:EQ1:Corollary:lossless_expander} and \refP{Equation:EQ2:Corollary:lossless_expander}
we expect that the NNLS has roughly the same estimation errors as the NNLAD for $r=1$, i.e. the evenly distributed noise,
and significantly worse estimation errors for $r=0$, i.e. the peaky noise.
\subsubsection*{Quality of the Estimation Error for Varying Sparsity}
We fix the constants $r=1$, $N=1024$, $M=256$, $D=10$, $SNR=1000$ and vary the sparsity level $S\in\SetOf{64}$.
For each $S$ we repeat \thref{Experiment:test:1} $100$ times.
We plot the mean of the relative $\ell_1$-estimation error and the mean of the
logarithmic relative $\ell_1$-estimation error, i.e.
\begin{align}\notag
	\MeanN=\Mean{\frac{\norm{\xVec-\xVec^\#}_1}{\norm{\xVec}_1}}
	\TextAnd\MeanLN{1}=\Mean{10\LogBase{10}{\frac{\norm{\xVec-\xVec^\#}_1}{\norm{\xVec}_1}}}
\end{align}
over the sparsity. The result can be found in \refP{Figure:Numcerics:l1_sphere_noise1} and
\refP{Figure:Numcerics:l1_sphere_noise2}.
\begin{figure}[ht]
    \begin{subfigure}[c]{0.5\textwidth}
		\includegraphics[scale=0.5]{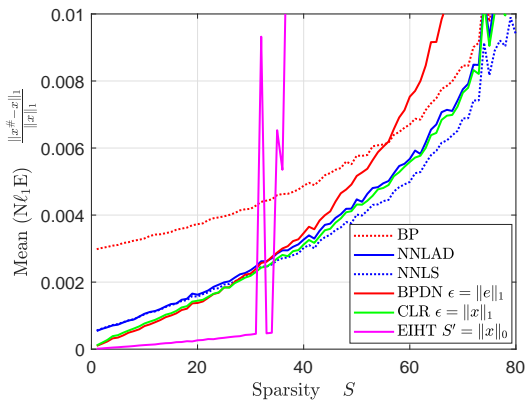}
		\subcaption{\label{Figure:Numcerics:l1_sphere_noise1}
			NNLAD has almost the same performance as CLR/BPDN. EIHT fails for moderate $S$.}
    \end{subfigure}
    \begin{subfigure}[c]{0.5\textwidth}
		\includegraphics[scale=0.5]{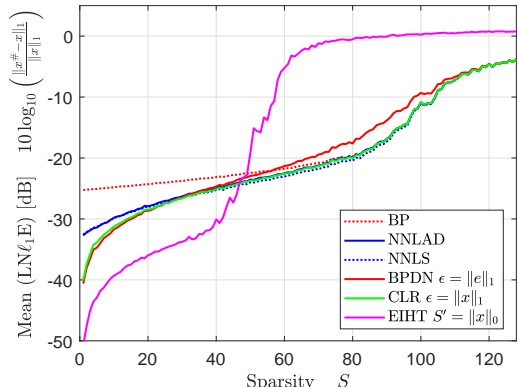}
		\subcaption{\label{Figure:Numcerics:l1_sphere_noise2}
			NNLAD and NNLS perform roughly the same.}
    \end{subfigure}
    \caption{Performance of NNLAD for noise with even mass noise and varying sparsity of the signal.}
\end{figure}
\\
For $S\geq 30$ the estimation error of the EIHT randomly peaks high.
We deduce that the EIHT fails to recover the signal reliably for $S\geq 30$, while the NNLAD and other convex decoders
succeed. This is not surprising, since by \cite[Theorem~13.15]{Introduction_CS} the EIHT obeys a robust recovery guartanee for
$S$-sparse signals, whenever $\AMat$ is the random wak matrix of a $\left(3S,D,\theta'\right)$-lossless expander
with $\theta'<\frac{1}{12}$. This is significantly stronger than the $\left(2S,D\theta\right)$-lossless expander property
with $\theta<\frac{1}{6}$ required for a null space property.
It might also be that the null space property is more likely than the lossless expansion property
similar to the gap between $\ell_2$-restricted isometry property
and null space property \cite{Gap}.
However, if the EIHT recovers a signal, it recovers it significantly better than any convex method.
This might be the case, since the originally generated signal is indeed from $\Sigma_S$, which is being enforced by
the hard thresholding of the EIHT, but not by the convex decoders.
This suggests that it might be useful to consider using thresholding on the output of any convex decoder to increase the accuracy if the orignal signal is indeed sparse and not only compressible.
For the remainder of this subsection we focus on convex decoders.\\
Contrary to our expectation the BPDN achieves worse estimation errors than all other convex decoders for $S\geq 60$, even worse
than the BP. The authors have no explanation for this phenomenon.
Apart from that we observe that the CLR and BP indeed perform as respectively best and worst case benchmark.
However, the difference between BP and CLR becomes rather small for high $S$.
We deduce that tuning becomes less important near the optimal
sampling rate.\\
The NNLAD, NNLS and CLR perform roughly the same.
This is quite strong, since BPDN and CLR are optimally tuned using
unknown prior information.
As expected the NNLS performs roughly the same as the NNLAD, see \refP{Table:main_results:table_of_advantages}.
However, this is the result of the noise distribution for $r=1$.
We repeat \thref{Experiment:test:1} with the same constants, but $r=0$,
i.e. $\eVec$ is a unit vector scaled by $\pm\frac{\norm{\AMat\xVec}_1}{SNR}$.
We plot the mean of the relative $\ell_1$-estimation error and the mean of the
logarithmic relative $\ell_1$-estimation error over the sparsity.
The result can be found in \refP{Figure:Numcerics:l0_sphere_noise1} and \refP{Figure:Numcerics:l0_sphere_noise2}.
\begin{figure}[ht]
    \begin{subfigure}[c]{0.5\textwidth}
		\includegraphics[scale=0.5]{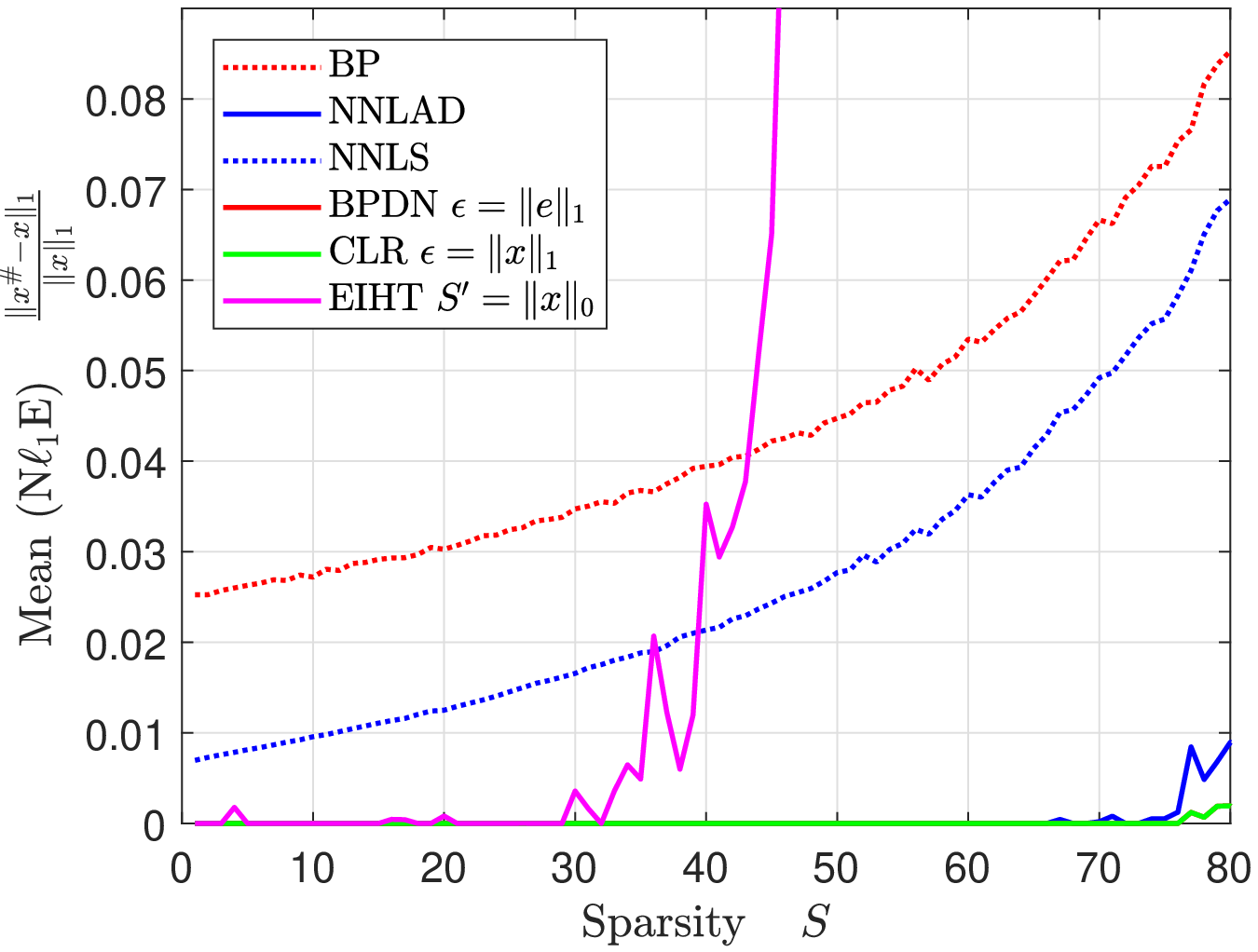}
		\caption{\label{Figure:Numcerics:l0_sphere_noise1}
			The NNLS does not fail, but performs bad.}
    \end{subfigure}
    \begin{subfigure}[c]{0.5\textwidth}
		\includegraphics[scale=0.5]{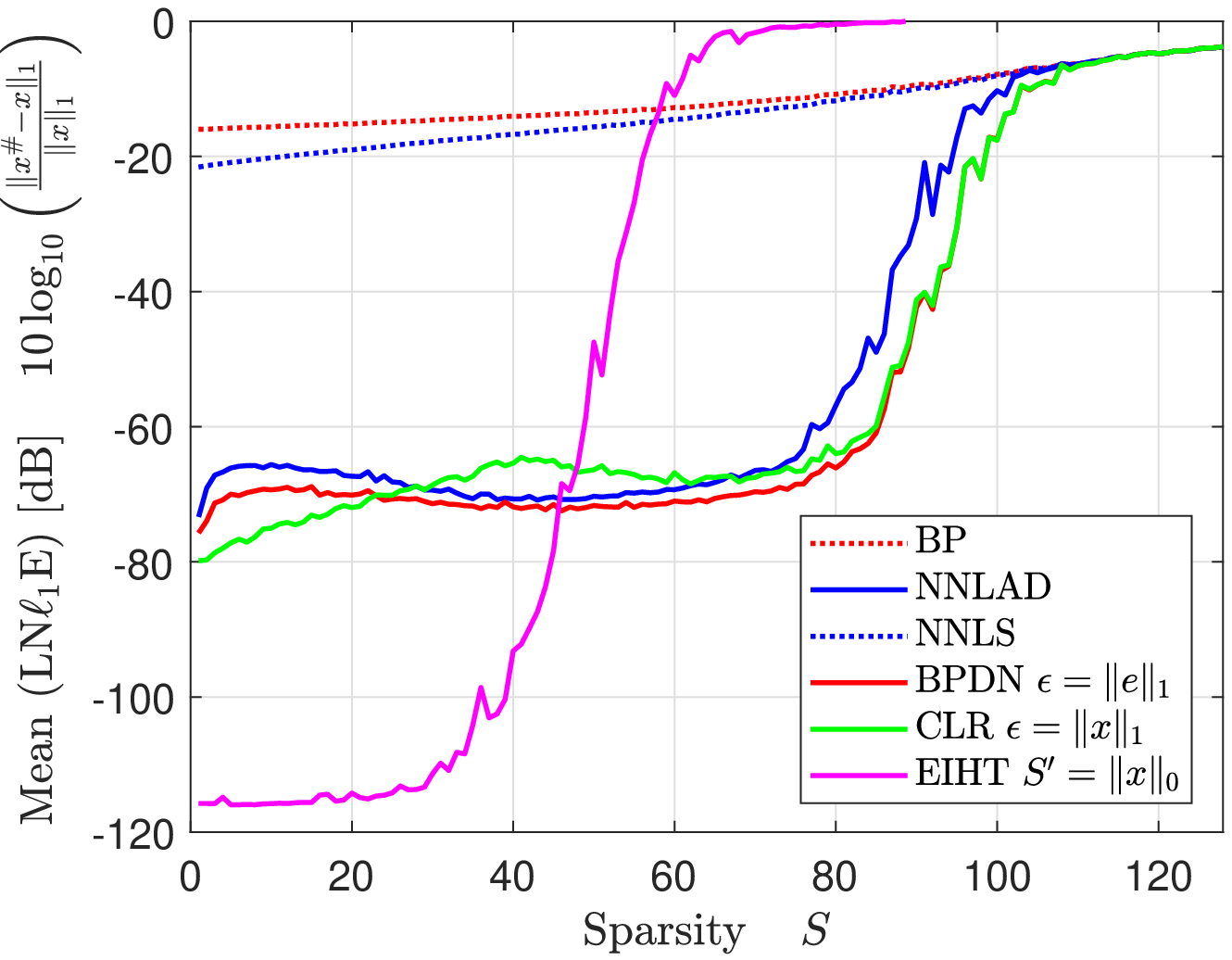}
		\caption{\label{Figure:Numcerics:l0_sphere_noise2}
        	The NNLS and NNLAD differ strongly.}
    \end{subfigure}
    \caption{Performance of NNLAD for noise with peaky mass and varying sparsity of the signal.}
\end{figure}
\\
We want to note that similarly to \refP{Figure:Numcerics:l1_sphere_noise1} the EIHT works only unreliably for $S\geq 30$.
Even though the mean of the logarithmic relative $\ell_1$-estimation error of NNLS is worse
than the one of EIHT for $30\leq S\leq 60$, the NNLS does not fail but only approximates with a weak error bound. 
As the theory suggests, the NNLS performs significantly worse than the NNLAD, see \refP{Table:main_results:table_of_advantages}.
It is worth to mention, that the estimaton errors of NNLS seem to be bounded by the estimation errors of BP.
This suggests that $\AMat$ obeys a $\ell_1$ quotient property, that bounds the estimation error of any instance optimal decoder,
see \cite[Lemma~11.15]{Introduction_CS}.
\subsubsection*{Noise-Blindness}
\thref{Theorem:NNLDMinimizer} states that the NNLAD has an error bound similarly to the optimally tuned CLR and BPDN.
Further, by \refP{Equation:EQ1:Corollary:lossless_expander} the ratio
\begin{align}\notag
	\frac{\norm{\xVec-\xVec^\#}_1}{\norm{\eVec}_1\norm{\xVec}_1}
	=\frac{\norm{\xVec-\xVec^\#}_1}{\norm{\eVec}_1}
\end{align}
should be bounded by some constant. To verify this,
we fix the constants $r=1$, $N=1024$, $M=256$, $D=10$, $S=32$ and vary the signal to noise ratio $SNR\in 10\SetOf{100}$.
For each $SNR$ we repeat \thref{Experiment:test:1} $100$ times.
We plot the mean of the logarithmic relative $\ell_1$-estimation error and the mean of the
ratio of relative $\ell_1$-estimation error and $\ell_1$-noise power, i.e.
\begin{align}\notag
	\MeanLN{1}=\Mean{10\LogBase{10}{\frac{\norm{\xVec-\xVec^\#}_1}{\norm{\xVec}_1}}}
	\TextAnd\MeanNperNP=\Mean{\frac{\norm{\xVec-\xVec^\#}_1}{\norm{\eVec}_1\norm{\xVec}_1}}
\end{align}
over the sparsity. The result can be found in \refP{Figure:Numcerics:Noise-Blindness_l1_sphere_noise1} and
\refP{Figure:Numcerics:Noise-Blindness_l1_sphere_noise2}.
\begin{figure}[ht]
    \begin{subfigure}[c]{0.5\textwidth}
        \includegraphics[scale=0.5]{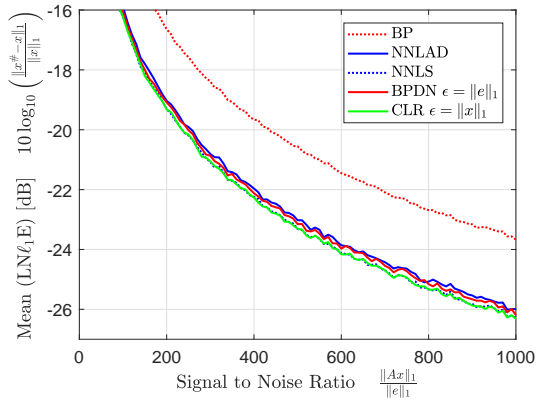}
		\caption{\label{Figure:Numcerics:Noise-Blindness_l1_sphere_noise1}
			The NNLAD and NNLS recover reliably for all signal to noise ratios.}
    \end{subfigure}
    \begin{subfigure}[c]{0.5\textwidth}
        \includegraphics[scale=0.5]{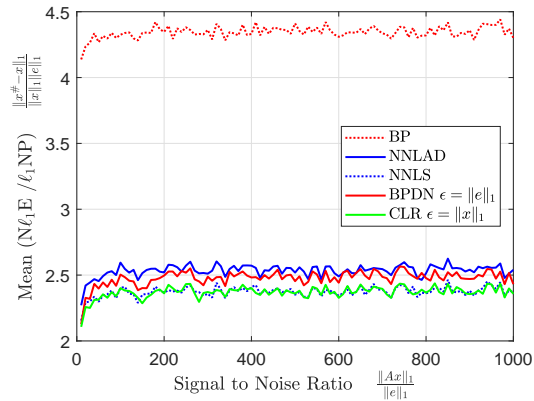}
		\caption{\label{Figure:Numcerics:Noise-Blindness_l1_sphere_noise2}
			The estimation error scales linearly with the noise power.}
    \end{subfigure}
    \caption{Performance of NNLAD for noise with even mass and varying noise power.}
\end{figure}
\newpage
The logarithmic relative $\ell_1$-estimation errors of the different decoders stay in a constant relation to each other
over the whole range of $SNR$. This relation is roughly the relation we can find in \refP{Figure:Numcerics:l1_sphere_noise2} for
$S=32$.
As expected the the ratio of relative $\ell_1$-estimation error and $\ell_1$-noise power stays constant independent
on the $SNR$ for all decoders.
We deduce that the NNLAD is noise-blind.
We repeat the experiment with $r=0$ and obtain \refP{Figure:Numcerics:Noise-Blindness_l0_sphere_noise1} and
\refP{Figure:Numcerics:Noise-Blindness_l0_sphere_noise2}.
\begin{figure}[ht]
    \begin{subfigure}[c]{0.5\textwidth}
        \includegraphics[scale=0.5]{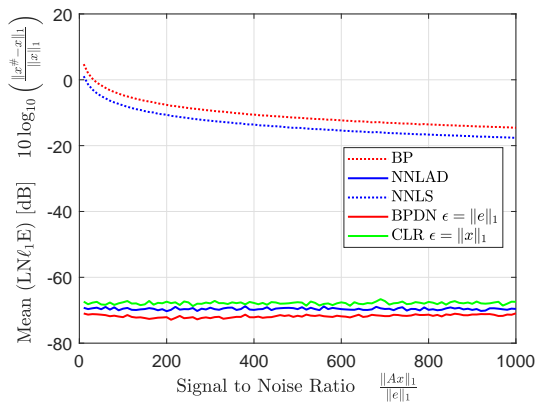}
		\caption{\label{Figure:Numcerics:Noise-Blindness_l0_sphere_noise1}
			The NNLAD outperforms the NNLS.}
    \end{subfigure}
    \begin{subfigure}[c]{0.5\textwidth}
        \includegraphics[scale=0.5]{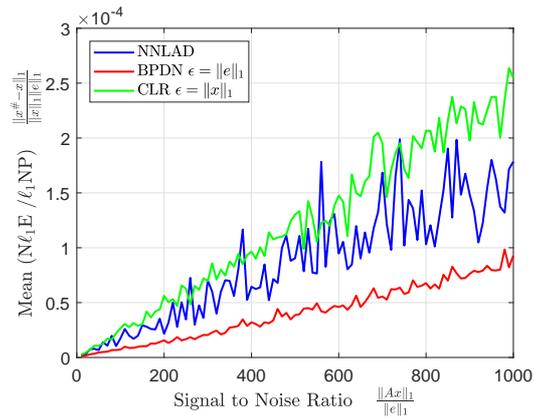}
		\caption{\label{Figure:Numcerics:Noise-Blindness_l0_sphere_noise2}
			The estimation error does not scale linearly with the noise power.}
    \end{subfigure}
    \caption{Performance of NNLAD for noise with peaky mass and varying noise power.}
\end{figure}
\\
Against our expectation, $\frac{\norm{\xVec-\xVec^\#}_1}{\norm{\xVec}_1}$ and not
$\frac{\norm{\xVec-\xVec^\#}_1}{{\norm{\xVec}_1}\norm{\eVec}_1}$
seems to be constant. Since $\frac{\norm{\xVec-\xVec^\#}_1}{\norm{\xVec}_1}\approx 1.0\cdot 10^{-7}$ is fairly small,
we suspect that this is the result of CVX reaching a tolerance parameter\footnote{The tolerance parameters of
CVX are the second and fourth root of the machine precision by default \cite{CVX1}, \cite{CVX2}.}
$\sqrt{eps}\approx 1.5\cdot 10^{-8}$
and terminating, while the actual optimizer might in fact be the original signal.
It is definitely noteworthy that even with the incredibly small signal to noise ration of $10$ the signal can be recovered 
by the NNLAD with an estimation error of $1.0\cdot 10^{-7}$ for this noise distribution.
\subsection{Decoding Complexity}
\subsubsection*{NNLAD vs iterative methods}
To investigate the convergence rates of the NNLAD as proposed in
\refP{Section:APP},
we compare it to different types of decoders when $\eVec=0$.
There are some sublinear time recovery methods for
lossless expander matrices including \cite{non-neg_BP}\cite[Section~13.4]{Introduction_CS}.
These are, as the name suggests,
significantly faster than the NNLAD.
These, as several other greedy methods \cite{combinatorial_Jafarpour}\cite{combinatorial_Xu}\cite{non-neg_BP}\cite{combinatoril_l1NSP_khaje}\cite[Section~13.3]{Introduction_CS},
rely on a strong lossless expansion property.
As a representative of all greedy and sublinear time methods
we will consider the EIHT,
which has a linear convergence rate $\OrderOf{c^{-k}}$ towards
the signal and robust recovery guarantees \cite[Theorem~13.15]{Introduction_CS}.
The EIHT also represents a best case benchmark.
As a direct competitor we consider the NNLS implemented by the methods of \cite{NNLS_first_order}
\footnote{This was the fastest method found by the authors.
	Other possibilities would be \cite[Algorithm~2]{APP}, \cite{NNLS_fista}.}
,  which has a convergence rate of $\OrderOf{k^{-2}}$ towards the optimal objective value.
\cite{NNLS_first_order} can also be used to calculate the
least shrinkage and selection operator.
However, calculating the projection onto the $\ell_1$-ball in $\mathbb{R}^N$,
is computationally slightly more complex than the projection onto $\mathbb{R}_+^N$.
Thus the NNLS will also be a lower bound for the LASSO.
As a worst case benchmark we consider a simple projected subgradient implementation of NNLAD using the Polyak step size,
i.e.
\begin{align}\tag{NNLAD Subgrad}
	\xVec^{k+1}:=\ProjToSet{\mathbb{R}^N_+}{\xVec^{k}
		-\frac{\norm{\AMat\xVec^k-\yVec}_1}{\norm{\AMat^T\sgn{\AMat\xVec^k-\yVec}}_2^2}\AMat^T\sgn{\AMat\xVec^k-\yVec}},
\end{align}
which has a convergence rate of $\OrderOf{k^{-\frac{1}{2}}}$ towards the optimal objective value
\cite[Section~7.2.2~\&~Section~5.3.2]{polyak_book}
\cite[Section~6]{subgradient_methods}. We will always initialize all iterated methods by zero vectors.
The EIHT will always use the parameter $S'=\norm{\xVec}_0$, the NNLAD $\sigma=\tau=0.99\norm{\AMat}_{2\rightarrow2}^{-1}$
and the NNLS the parameters $s=0.99\norm{\AMat}_{2\rightarrow2}^{-2}$ and $\alpha=3.01$, see \cite{NNLS_first_order}.
Parameters that can be computed from $\AMat$, will be calculated before the timers start. This includes
the adjacency structure of $\AMat$ for the EIHT, $\sigma$, $\tau$ for NNLAD, $s$, $\alpha$ for NNLS,
since these are considered to be a part of the decoder.
We will do the following experiment multiple times:
\begin{Experiment}\label{Experiment:test:2}
	\hspace{1pt}
	\begin{itemize}
		\item[1.] If $r=1$, generate a measurement matrix $\AMat\in\left\{0,D^{-1}\right\}^{M\times N}$
			as a uniformly at random drawn $D$-LRBG.\\
			If $r=2$, draw each component $A_{m,n}$ of the measurement matrix independent and uniformly at random
			from $\left\{0,1\right\}$, i.e. as $0/1$ Bernoulli random variables.
		\item[2.] Generate a signal $\xVec$ uniformly at random from $\Sigma_S\cap\mathbb{R}_+^N\cap\mathbb{S}_r^{N-1}$.
		\item[3.] Define the observation $\yVec:=\AMat\xVec$.
		\item[4.] For each iterative method calculate the sequence of estimators $\xVec^k$ for all $k\leq 20000$
			and collect the relative estimation errors
			$\frac{\norm{\xVec^k-\xVec}_1}{\norm{\xVec}_1}$,
			the relative norms of the residuals
			$\frac{\norm{\AMat\xVec^k-\yVec}_1}{\norm{\yVec}_1}$
			and the time to calculate the first $k$ iterations.
	\end{itemize}
\end{Experiment}
For $r=2$ this represents a biased sub-gaussian random ensemble \cite{NNLS_first} with optimal recovery guarantees for
the NNLS. For $r=1$ this represents a $D$-LRBG random ensemble with optimal recovery guarantees for the NNLAD.
We fix the constants $r=1$, $N=1024$, $M=256$, $S=16$, $D=10$ and repeat \thref{Experiment:test:2} $100$ times.
We plot the mean of the logarithmic relative $\ell_1$-estimation error
and the mean of the relative $\ell_1$-norm of the residual, i.e.
\begin{align}\notag
	\MeanLN{1}=\Mean{10\LogBase{10}{\frac{\norm{\xVec^k-\xVec}_1}{\norm{\xVec}_1}}}
	\TextAnd\MeanLNR{1}=\Mean{10\LogBase{10}{\frac{\norm{\AMat\xVec^k-\yVec}_1}{\norm{\yVec}_1}}}
\end{align}
over the sparsity and the time.
The result can be found in
\refP{Figure:Numcerics:computational_complexity:LRBG:1} and 
\refP{Figure:Numcerics:computational_complexity:LRBG:2}.
\begin{figure}[ht]
	\begin{subfigure}[c]{0.5\textwidth}
		\includegraphics[scale=0.5]{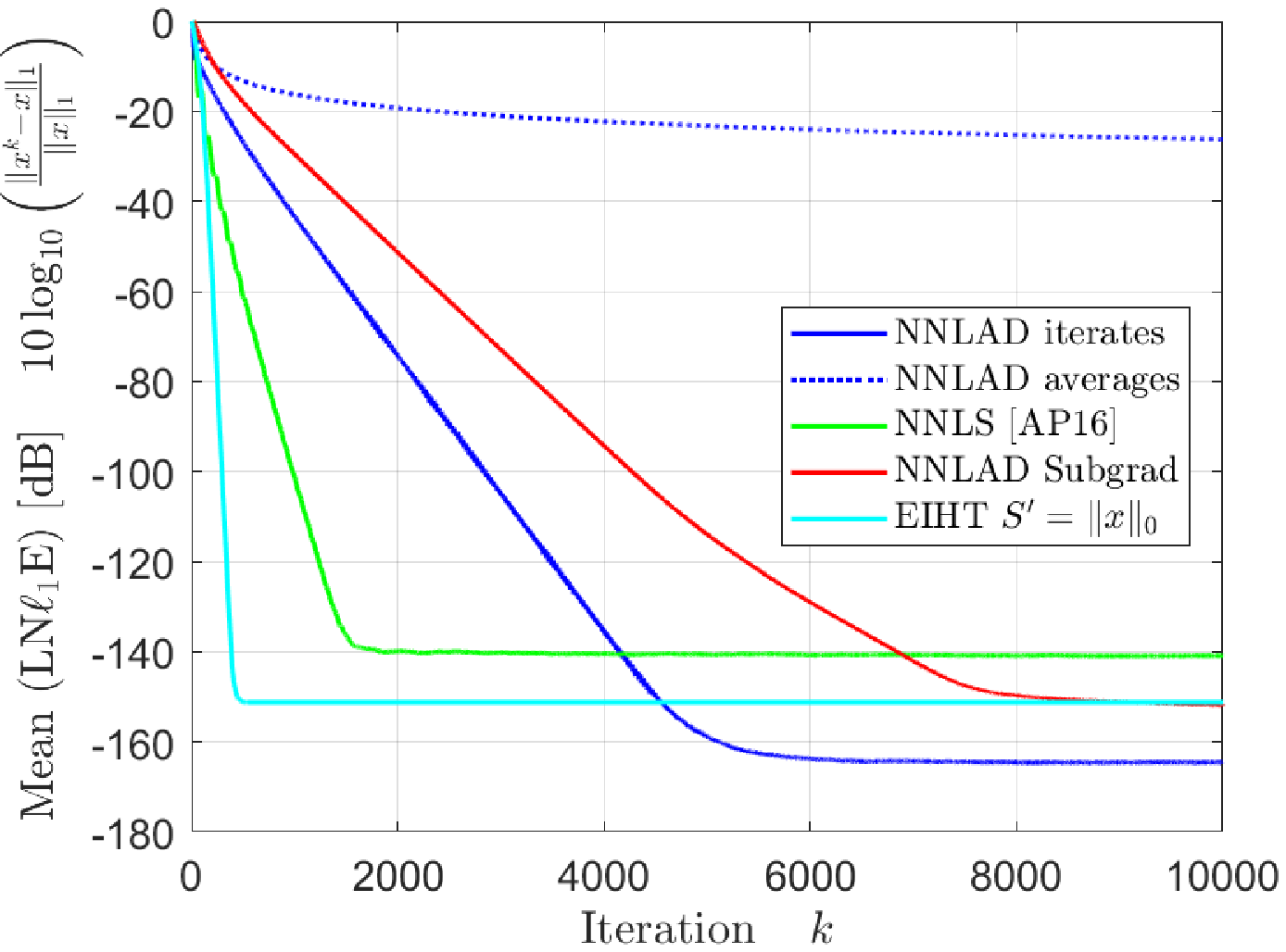}
	\end{subfigure}
	\begin{subfigure}[c]{0.5\textwidth}
		\includegraphics[scale=0.5]{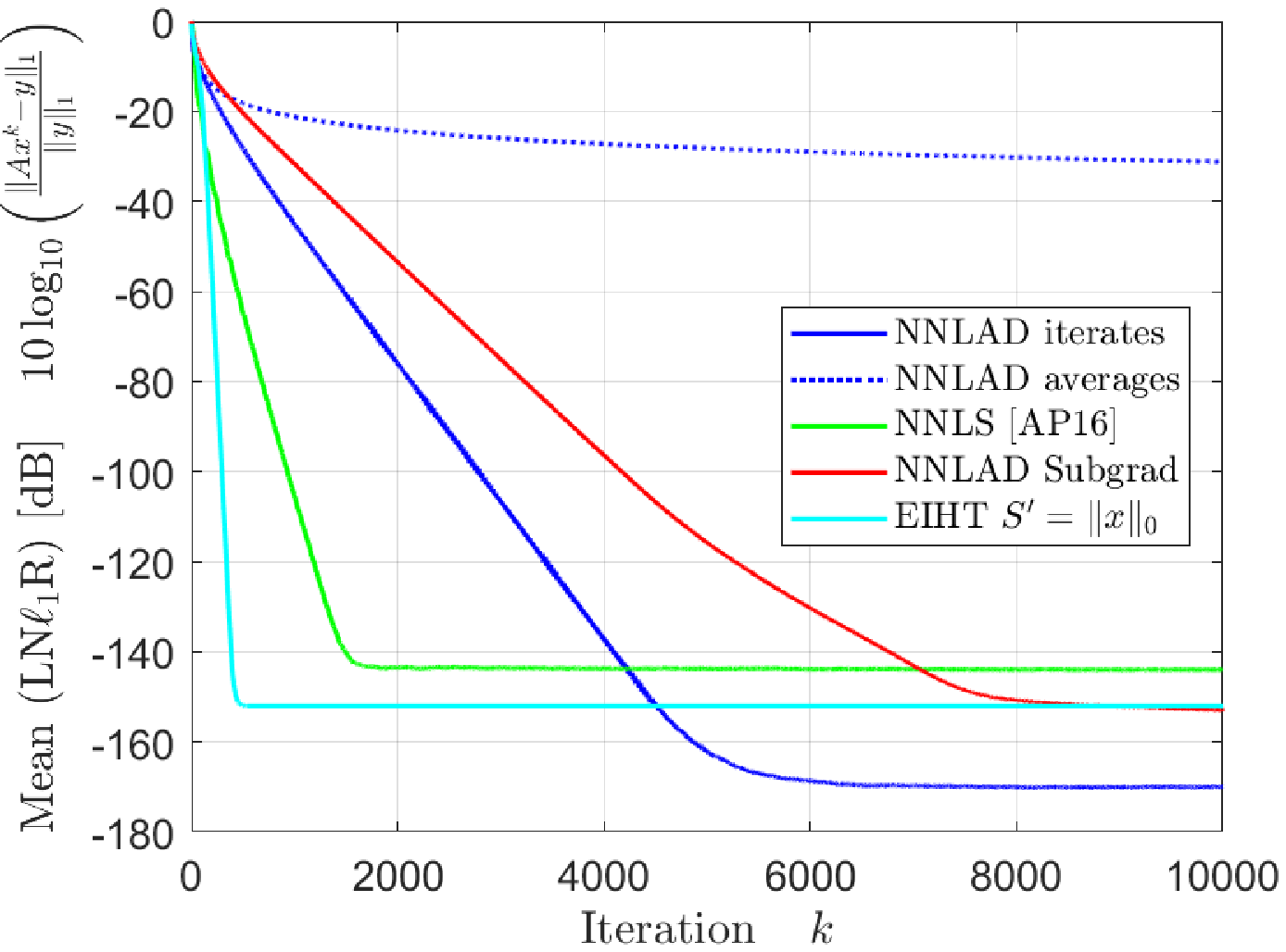}
	\end{subfigure}
	\caption{\label{Figure:Numcerics:computational_complexity:LRBG:1}
			Convergence rates of certain iterated methods with respect to the number of iterations.}
\end{figure}
\begin{figure}[ht]
	\begin{subfigure}[c]{0.5\textwidth}
		\includegraphics[scale=0.5]{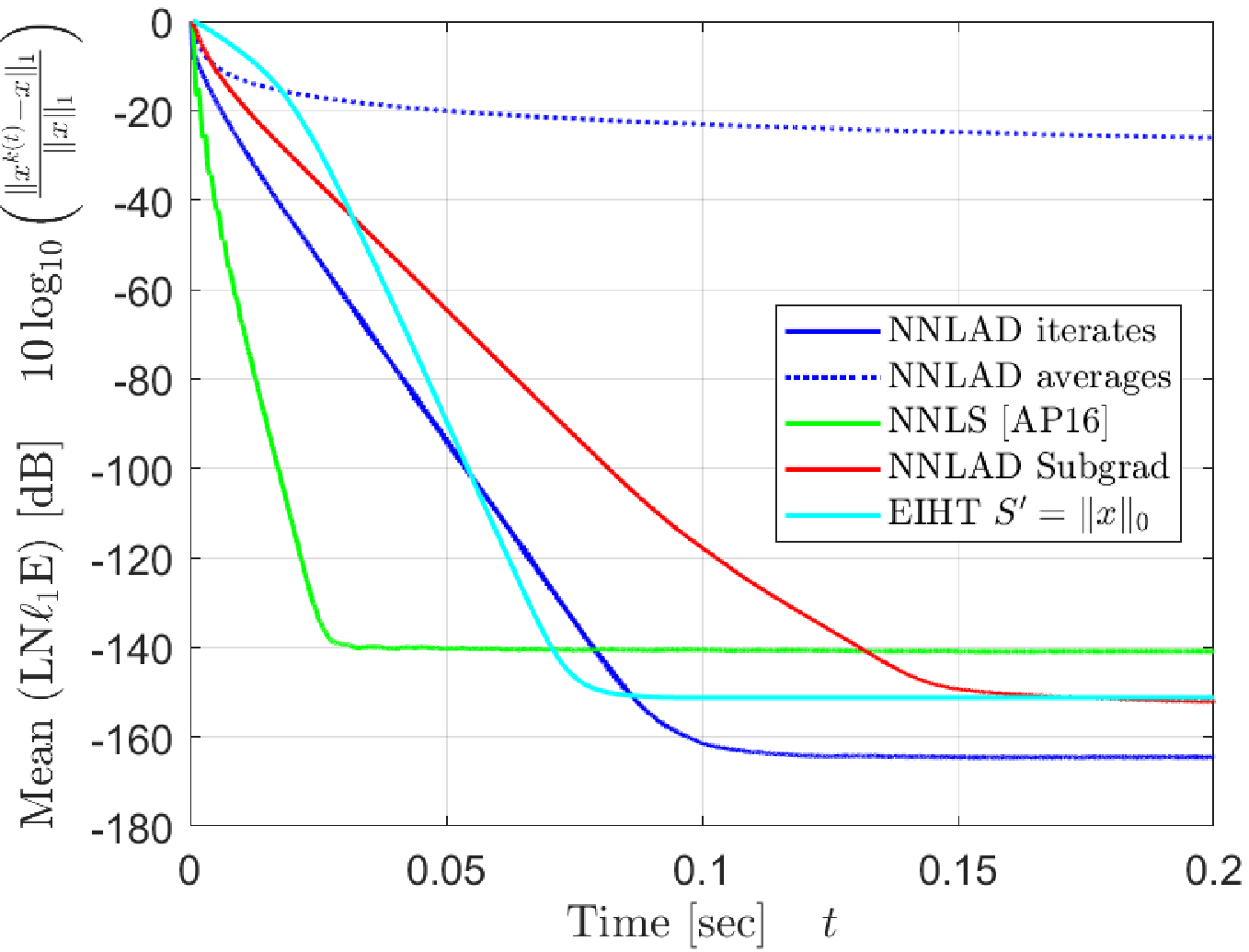}
	\end{subfigure}
	\begin{subfigure}[c]{0.5\textwidth}
		\includegraphics[scale=0.5]{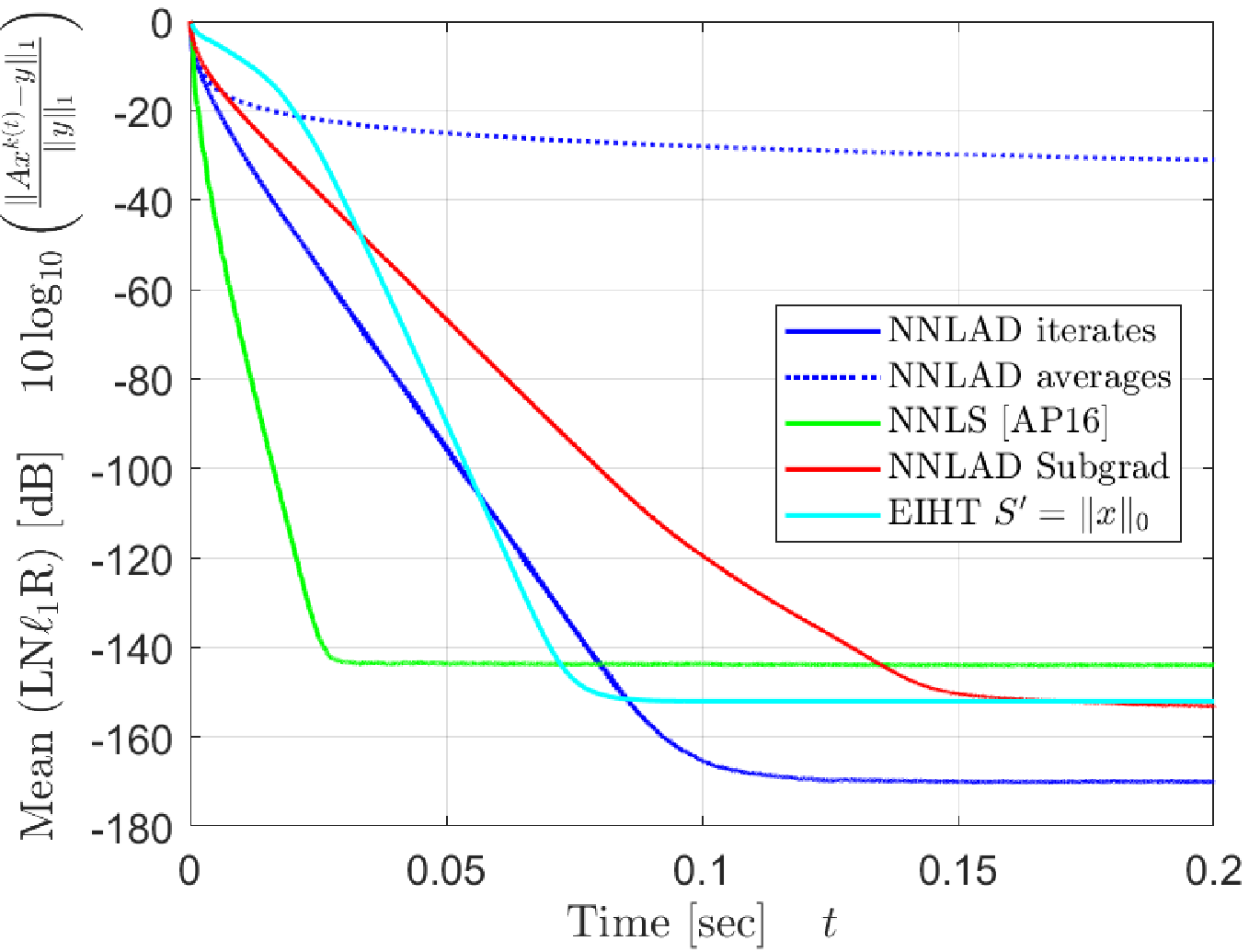}
	\end{subfigure}
	\caption{\label{Figure:Numcerics:computational_complexity:LRBG:2}
			Convergence rates of certain iterated methods with respect to the time.}
\end{figure}
\newpage
The averages of NNLAD converge significantly slower than the iterates, even though we lack a convergence rate for the iterates.
We deduce that one should always use the iterates of NNLAD to recover a signal.
Surprisingly, the averages converge even slower than the subgradient method. However, this is not because 
the averages converge slow, but rather because the subgradient method and all others converges faster than expected.
In particular, the NNLAD iterates, EIHT and the NNLS all converge linearly towards the optimal objective value
and towards the signal. Even the subgradient method converges almost linearly.
We deduce that the NNLS is the fastest of these methods if $\AMat$ is a $D$-LRBG.\\
Apart from a constant the NNLAD iterates, EIHT and NNLS converge in the same order.
However, this behavior does not hold if we consider a different distribution for $\AMat$ as one can verify
by setting each component $A_{m,n}$ as independent $0/1$ Bernoulli random variables.
While EIHT has better iterations compared to the NNLS, it still takes more time to achieve the same
estimation errors and residuals.
We plot the mean of the time required to calculate the first $k$ iterations in
\refP{Figure:Numcerics:computational_complexity:LRBG_time_per_iteration}.
\begin{figure}[ht]
		\centering
        \includegraphics[scale=0.5]{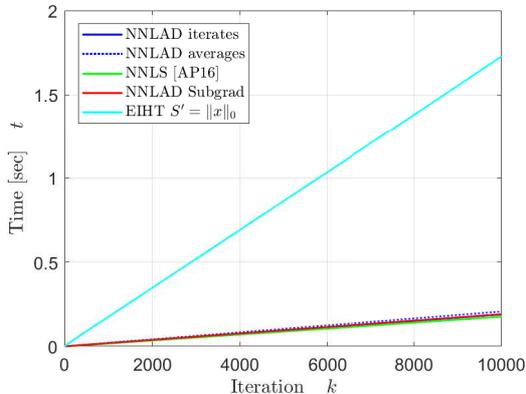}
		\caption{\label{Figure:Numcerics:computational_complexity:LRBG_time_per_iteration}
			Time required to perform iterations of certain iterated methods.}
\end{figure}
\\
The EIHT requires roughly $6$ times as long as any other method to calculate each iteration.
All methods but the EIHT can be implemented with only two matrix vector multiplications, namely once by $\AMat$
and once by $\AMat^T$. Both of these requires roughly $2DN$ floating point operations.
Hence, each iteration requires $\OrderOf{4DN}$ floating point operations.
The EIHT only calculates one matrix vector multiplication, but also the median.
This calculation is significantly slower than a matrix vector multiplication.
For every $n\in\SetOf{N}$ we need to order a vector with $D$ elements, which can be performed in
$\OrderOf{D\Log{D}}$. Hence, each iteration of EIHT requires $\OrderOf{DN\Log{D}}$ floating point operations,
which explains why the EIHT requires significantly more time for each iteration.\\
As we have seen the NNLS is able to recover signals faster than any other method, however
it also only obeys sub-optimal robustness guarantees for uniformly at random
chosen $D$-LRBG as we have seen in
\refP{Figure:Numcerics:Noise-Blindness_l0_sphere_noise1}. We ask ourself whether
or not the NNLS is also faster with a more natural measurement scheme, i.e. if $A_{m,n}$ are
independent $0/1$ Bernoulli random variables.
We repeat \thref{Experiment:test:2} $100$ times with $r=2$ for the NNLS and $r=1$ for the other methods.
We again plot the mean of the logarithmic relative $\ell_1$-estimation error
and the mean of the relative $\ell_1$-norm of the residual
in \refP{Figure:Numcerics:computational_complexity:LRBGvsSubGauss:1}
and \refP{Figure:Numcerics:computational_complexity:LRBGvsSubGauss:2}.
\begin{figure}[ht]
	\begin{subfigure}[c]{0.5\textwidth}
		\includegraphics[scale=0.5]{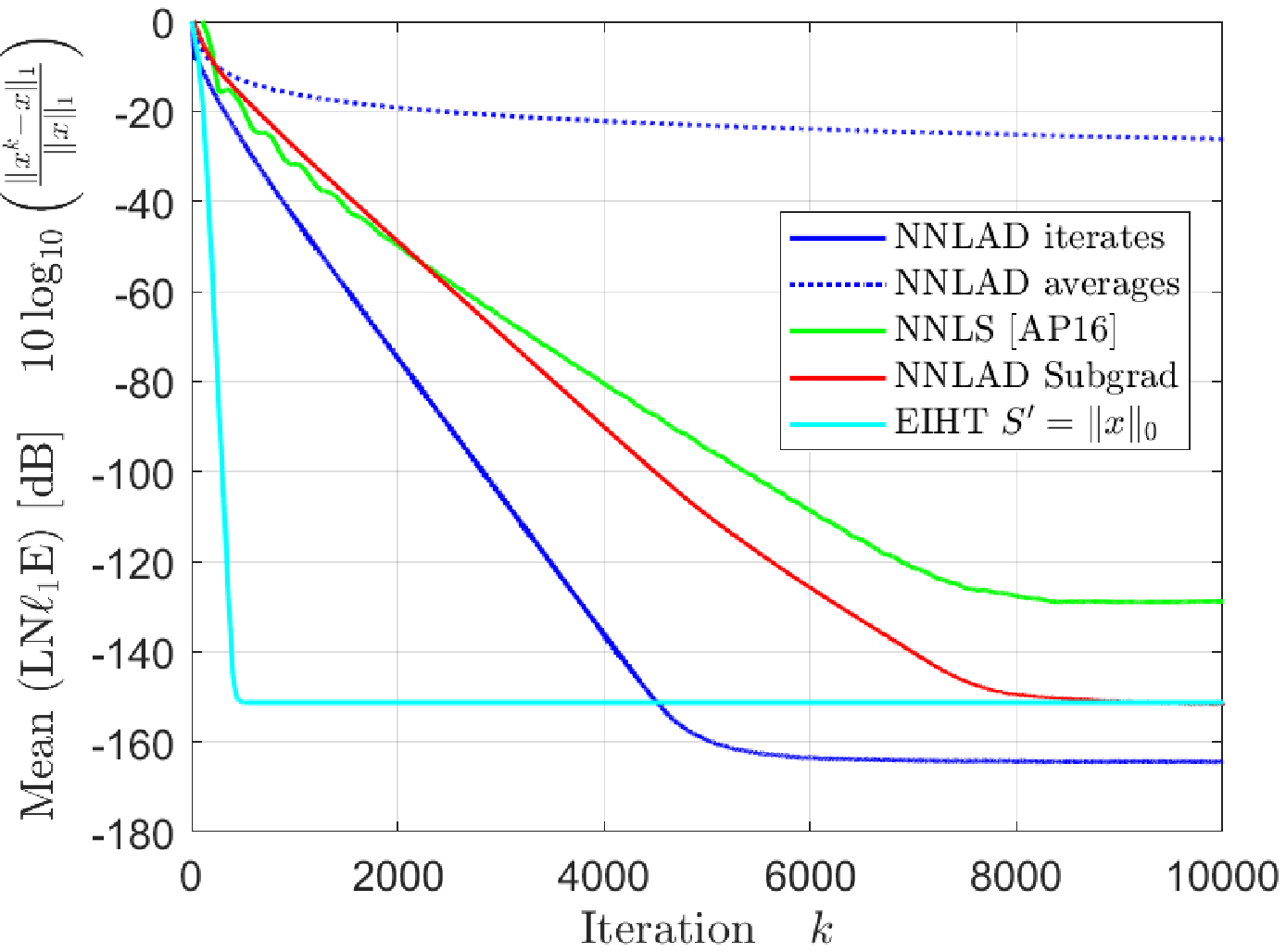}
	\end{subfigure}
	\begin{subfigure}[c]{0.5\textwidth}
		\includegraphics[scale=0.5]{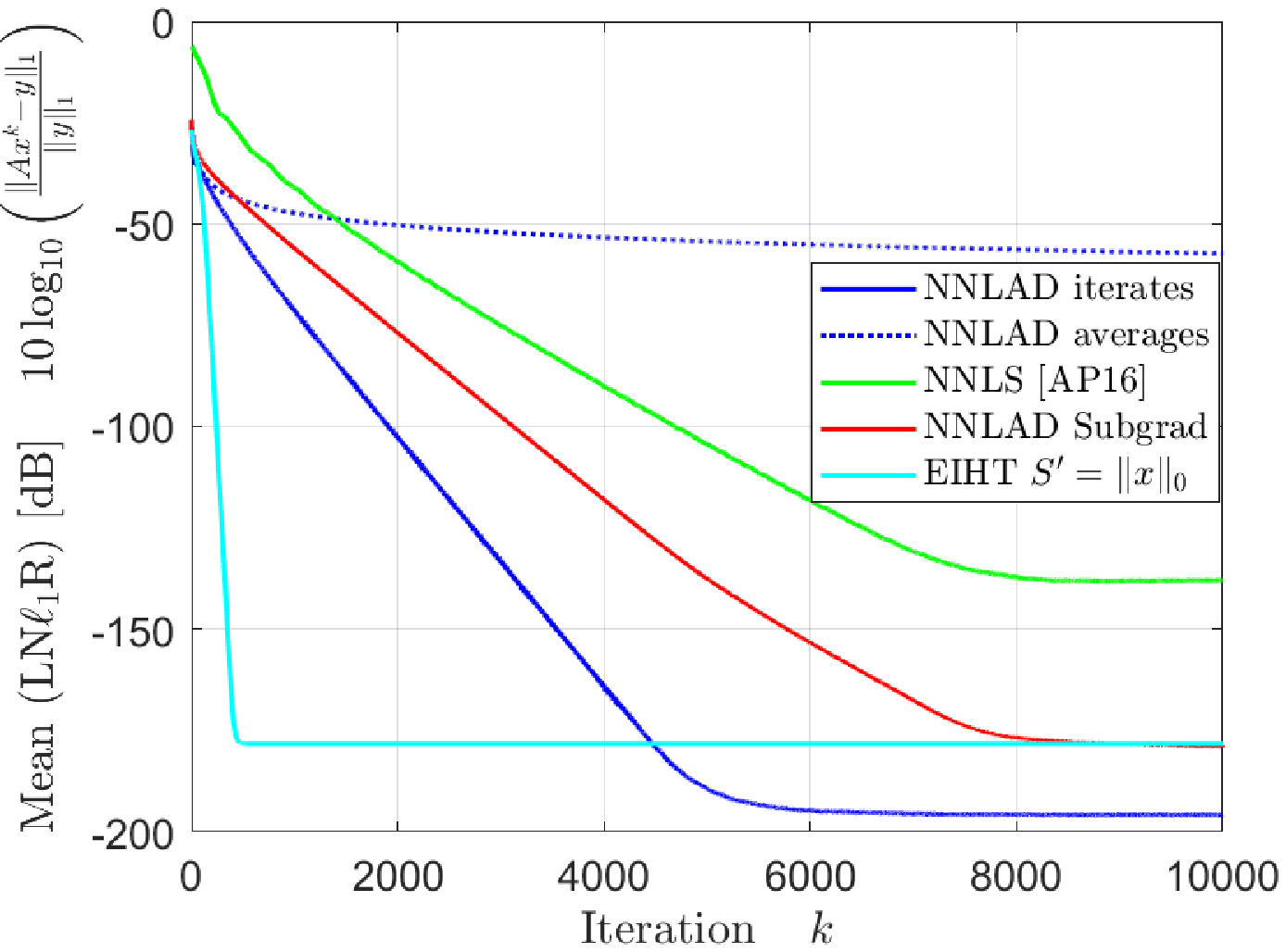}
	\end{subfigure}
	\caption{\label{Figure:Numcerics:computational_complexity:LRBGvsSubGauss:1}
			Convergence rates of certain iterated methods with respect to the number of iterations.
			$\AMat$ is Bernoulli for NNLS and $D$-LRBG for the others.}
\end{figure}
\begin{figure}[ht]
	\begin{subfigure}[c]{0.5\textwidth}
		\includegraphics[scale=0.5]{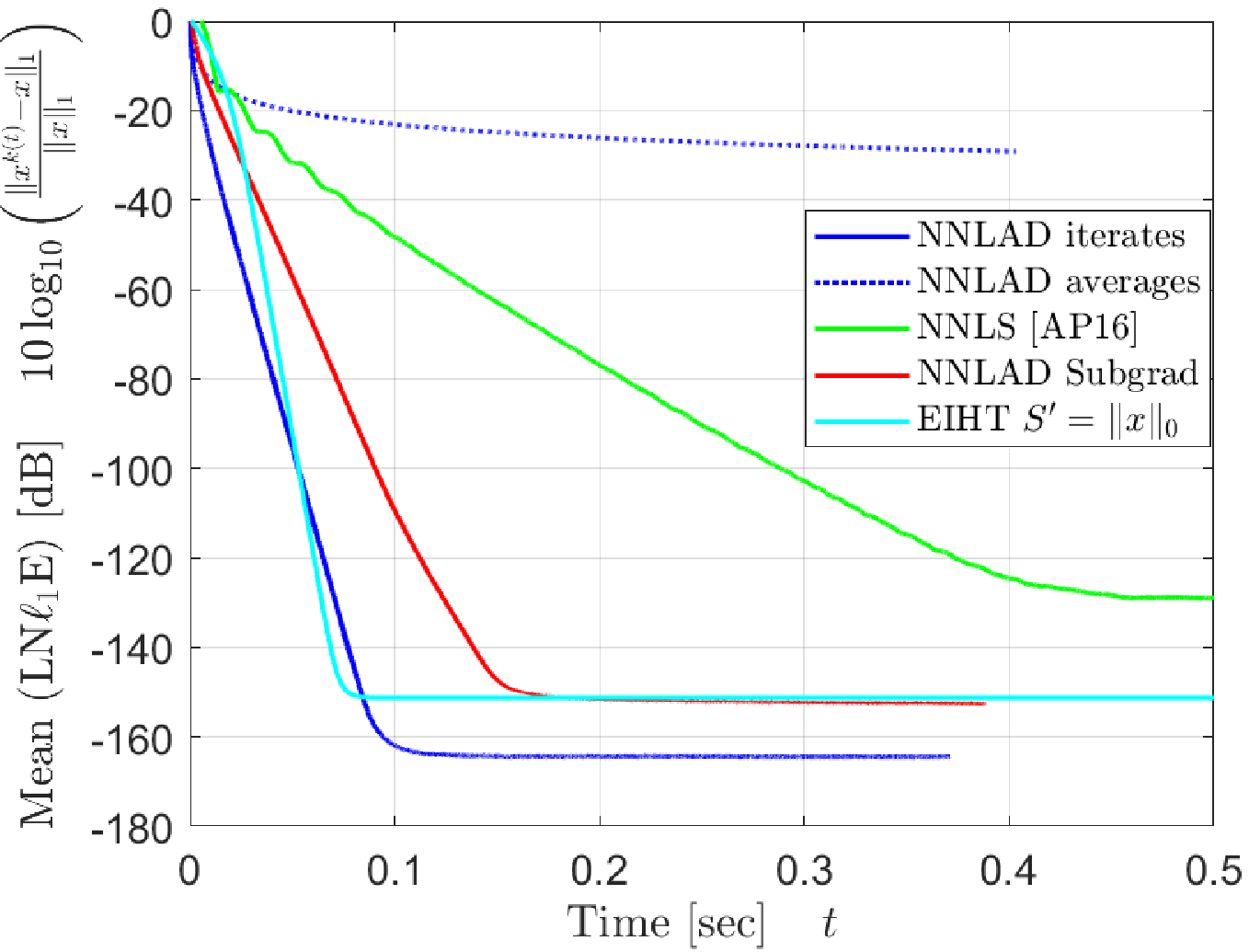}
	\end{subfigure}
	\begin{subfigure}[c]{0.5\textwidth}
		\includegraphics[scale=0.5]{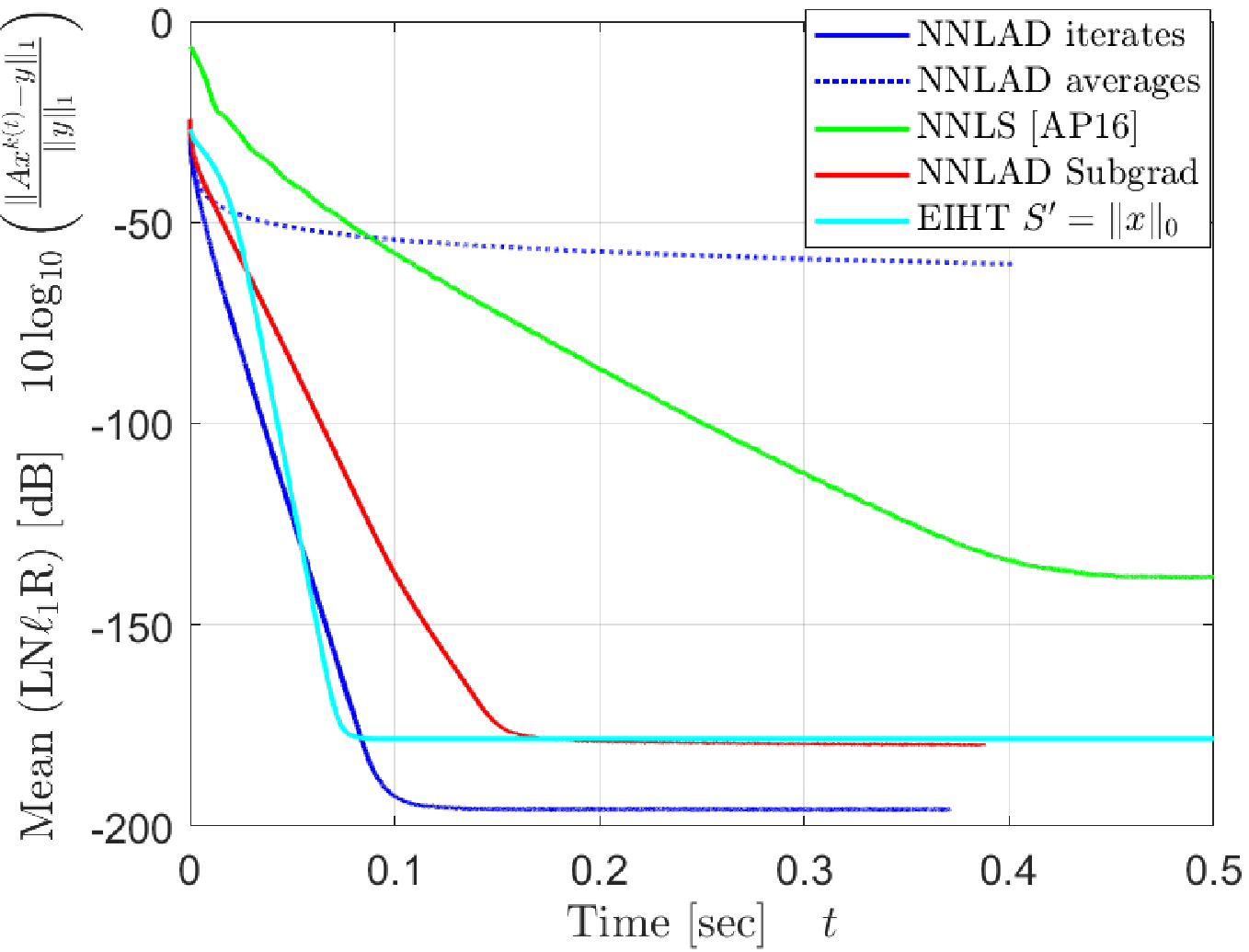}
	\end{subfigure}
	\caption{\label{Figure:Numcerics:computational_complexity:LRBGvsSubGauss:2}
			Convergence rates of certain iterated methods with respect to the time.
			$\AMat$ is Bernoulli for NNLS and $D$-LRBG for the others.}
\end{figure}
\newpage
The NNLAD and the EIHT converge to the solution with roughly the same time. Even the subgradient implementation of the NNLAD
recovers a signal in less time than the NNLS. Further the convergence of NNLS does not seem to be linear anymore.
We deduce that sparse structure of $\AMat$ has a more significant influence on the decoding time than
the smoothness of the data fidelity term.
Also we deduce that even the subgradient method is a viable choice to recover a signal.
\subsubsection*{NNLAD vs SPGL1}
As a last test we compare the NNLAD to the SPGL1 \cite{SPGL1}\cite{SPGL1_Site} toolbox for matlab.
\begin{Experiment}\label{Experiment:test:3}
	\hspace{1pt}
	\begin{itemize}
		\item[1.] Generate the measurement matrix $\AMat\in\left\{0,D^{-1}\right\}^{M\times N}$ as a uniformly at random drawn $D$-LRBG.
		\item[2.] Generate the signal $\xVec$ uniformly at random from $\Sigma_S\cap\mathbb{R}_+^N\cap\mathbb{S}_r^{N-1}$.
		\item[3.] Define the observation $\yVec:=\AMat\xVec$.
		\item[4.] Use a benchmark decoder to calculate an estimator $\xVec^\#$ and collect
			the relative estimation errors\\
			$\frac{\norm{\xVec^\#-\xVec}_1}{\norm{\xVec}_1},\frac{\norm{\xVec^\#-\xVec}_2}{\norm{\xVec}_2}$
			and the time to calculate $\xVec^\#$.
		\item[5.] For each iterative method calculate iterations
			until $\frac{\norm{\xVec^k-\xVec}_1}{\norm{\xVec}_1}\leq\frac{\norm{\xVec^\#-\xVec}_1}{\norm{\xVec}_1}$
			and $\frac{\norm{\xVec^k-\xVec}_2}{\norm{\xVec}_2}\leq\frac{\norm{\xVec^\#-\xVec}_2}{\norm{\xVec}_2}$.
			Collect the time to perform these iterations.
			If this threshold can not be reached after $10^5$ iterations, the recovery failed
			and the time is set to $\infty$.
	\end{itemize}
\end{Experiment}
We again fix the dimension $N=1024$, $M=256$, $D=10$ and vary $S\in\SetOf{128}$. For both the BP implementation of SPGL1
and the LASSO implementation of SPGL1 we repeat \thref{Experiment:test:3} $100$ times for each $S$.
We plot the mean of the time to calculate the estimators and plot these over the sparsity
in \refP{Figure:Numcerics:computational_complexity:BP_SPGL1} and \refP{Figure:Numcerics:computational_complexity:LASSO_SPGL1}.
\begin{figure}[ht]
    \begin{subfigure}[c]{0.5\textwidth}
        \includegraphics[scale=0.5]{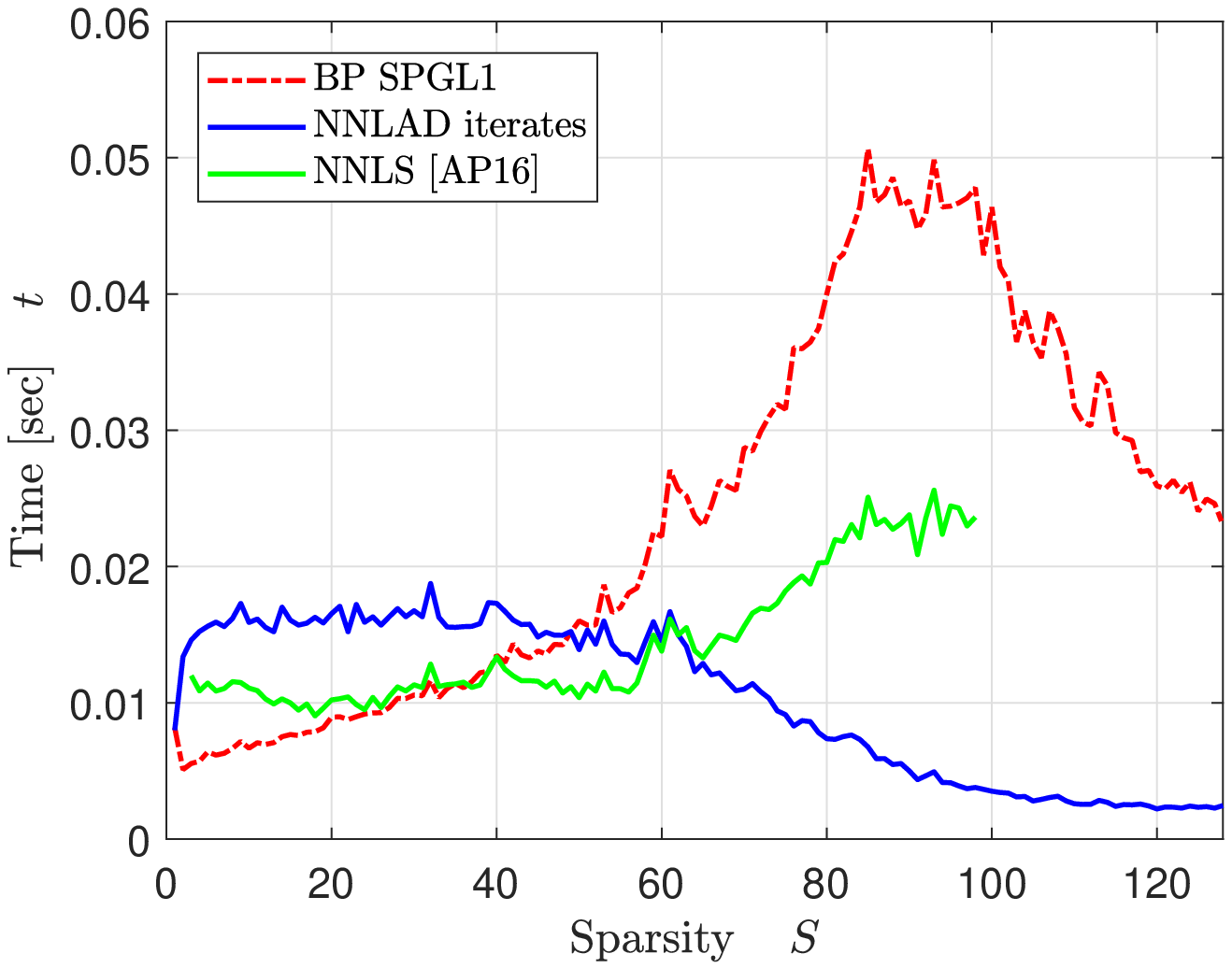}
		\caption{\label{Figure:Numcerics:computational_complexity:BP_SPGL1}
			The NNLAD is faster than the BP of SPGL1 for high $S$.}
    \end{subfigure}
    \begin{subfigure}[c]{0.5\textwidth}
        \includegraphics[scale=0.5]{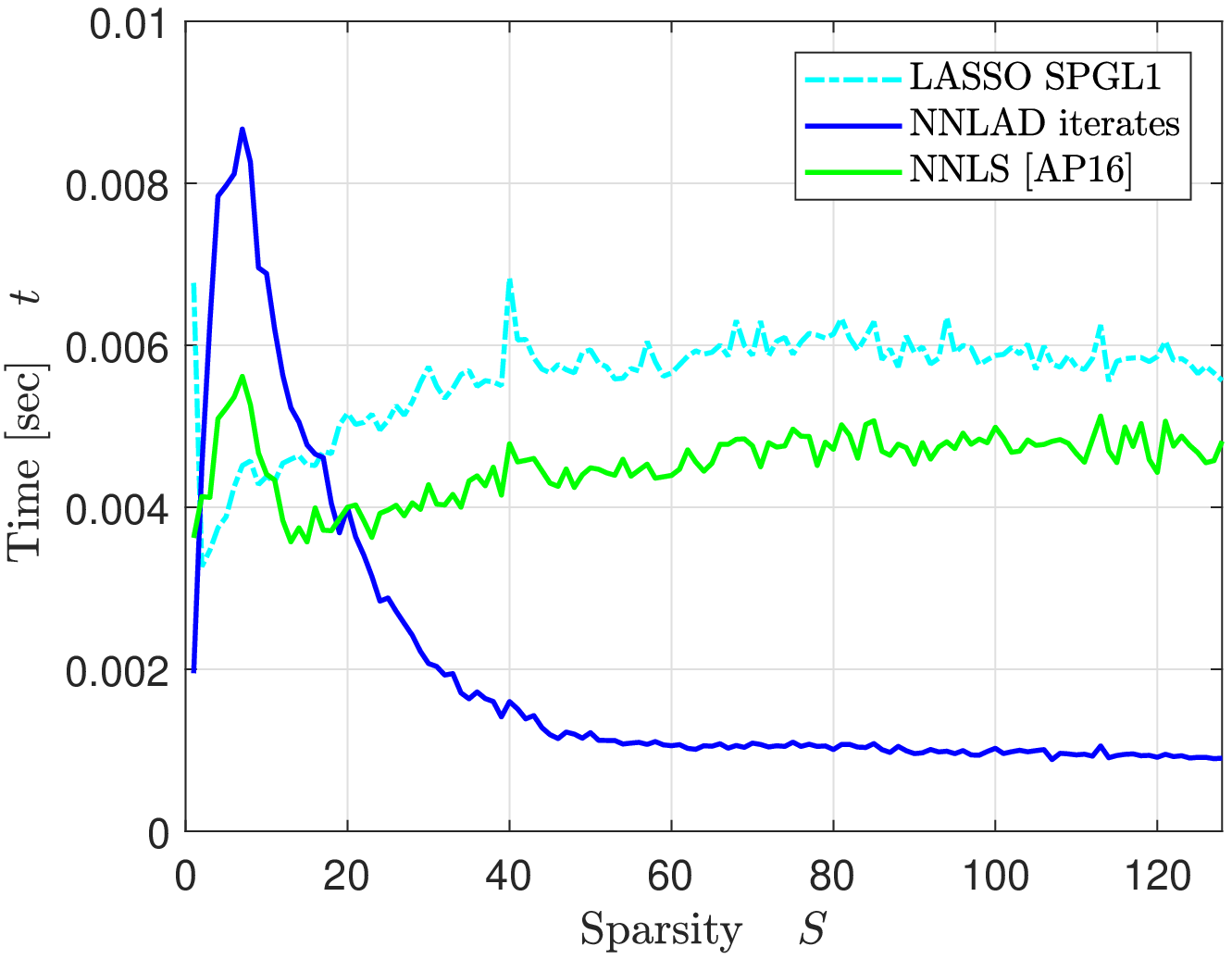}
		\caption{\label{Figure:Numcerics:computational_complexity:LASSO_SPGL1}
			The NNLAD is faster than the LASSO of SPGL1 for moderate $S$.}
    \end{subfigure}
    \caption{Time of the NNLAD and NNLS to approximate better than SPGL methods.}
\end{figure}
\newpage
The NNLAD implementation is slower than both SPGL1 methods for small $S$.
However, if we have the optimal number of measurements $M\in\mathcal{O}\left(S\Log{\frac{N}{S}}\right)$,
the NNLAD is faster than both SPGL1 methods.
\subsubsection*{Summary}
The implementation of NNLAD as presented in \thref{Algorithm:NNLAD_APP_no_average} is a reliable recovery method
for sparse non-negative signals. There are methods that might be faster, but these
either recover a smaller number of coefficients (EIHT, greedy methods)
or they obey sub-optimal recovery guarantees (NNLS).
The implementation is as fast as the commonly uses SPGL1 toolbox, but has the advantage that it requires no tuning
depending on the unknown $\xVec$ or $\eVec$.
Lastly, the NNLAD can handle peaky noise overwhelmingly good.
\subsection{Application for Viral Detection}\label{Subsection:applications_for_COVID-19}
With the outbreak and rapid spread of the COVID-19 virus we are in the need of testing
a lot of people for an infection.
Since we can only test a fixed number of persons in a given time,
the number of persons tested for the virus grows at most linearly.
On the other hand, models suggest that the number of possibly infected persons grows exponentially.
At some point, if that is not already the case, we will have a shortage of test kits
and we will not be able to test every person. It is thus desirable to test
as much persons with as few as possible test kits.\\
The field \emph{group testing} develops strategies to test groups of individuals instead of
individuals in order to reduce the amount of tests required to identify infected individuals.
The first advances in group testing were made in \cite{dorfman_original}. For a general overview
about group testing we refer to \cite{group_testing}.\\
The problem of testing a large group for a virus can be modeled as a compressed sensing problem
in the following way:
Suppose we want to test $N$ persons, labeled by $\SetOf{N}=\left\{1,\dots,N\right\}$,
to check whether or not they are affected by a virus.
We denote by $x_n$ the quantity of viruses
in the specimen of the $n$-th person.
Suppose we have $M$ test kits, labeled by $\SetOf{M}=\left\{1,\dots,M\right\}$.
By $y_m$ we denote the amount of viruses in the sample
of the $m$-th test kit.
Let $\AMat\in\left[0,1\right]^{M\times N}$.
For every $n$ we put a fraction of size $A_{m,n}$ of the specimen of the $n$-th person
into the sample for the $m$-th test kit.
The sample of the $m$-th test kit will then have the quantity of viruses
\begin{align*}
	\sum_{n\in\SetOf{N}}A_{m,n}x_n +e^{con}_m,
\end{align*}
where $e^{con}_m$ is the amount of viruses in the sample originating from a possible contamination
of the sample.
A quantitative reverse transcription polymerase chain reaction
estimates the quantity of viruses by $y_m$ with a small error
$e^{pcr}_m=y_m-\sum_{n\in\SetOf{N}}A_{m,n}x_n -e^{con}_m$.
After all $M$ tests we detect the quantity
\begin{align}
	\yVec=\AMat\xVec+\eVec,
\end{align}
where $\eVec=\eVec^{con}+\eVec^{pcr}$.
Since contamination of samples happens rarely, $\eVec^{con}$ is assumed to be
peaky in terms of \refP{Table:main_results:table_of_advantages},
while $\eVec^{pcr}$ is assumed to have even mass
but a small norm.
In total $\eVec$ is peaky.\\
Often each specimen is tested separately, meaning that $\AMat$ is the identity.
In particular, we need at least as much test kits as specimens.
Further, we estimate the true quantity of viruses $x_n$
by $x^\#_n:=y_n$, which results in the estimation error
$x^\#_n-x_n=e_n=e^{con}_n+e^{pcr}_n$.
Since the noise vector $\eVec$ is peaky,
some but few tests will be inaccurate and might result
in false positives or false negatives.\\
In general, only a fraction of persons is indeed affected by the virus.
Thus, we assume that $\norm{\xVec}_0\leq S$ for some small $S$.
Since the amount of viruses is a non-negative value, we also have $\xVec\geq 0$.
Hence, we can use the NNLR to estimate $\xVec$ and in particular
we should use the NNLAD due to the noise being peaky.
\thref{Corollary:lossless_expander} suggests to choose $\AMat$ as the random walk matrix of a
lossless expander or by \cite[Theorem~13.7]{Introduction_CS} to choose $\AMat$ as a
uniformly at random chosen $D$-LRBG.
Such a matrix $\AMat$ has non-negative entries and the column sums of $\AMat$ are
not greater than one. This is a necessary requirement since each column sum is the total amount of
specimen used in the test procedure.
Especially, a fraction of $D^{-1}$ of each specimen is used in exactly $D$ test kits.\\
By \thref{Corollary:lossless_expander} and
\cite[Theorem~13.7]{Introduction_CS} this allows us to reduce the number of test kits
required to $M\approx CS\Log{\ExpE\frac{N}{S}}$.
As we have seen in \refP{Figure:Numcerics:Noise-Blindness_l0_sphere_noise1} and
\refP{Figure:Numcerics:Noise-Blindness_l0_sphere_noise2}
we expect the NNLAD estimator to correct the errors from
$\eVec^{con}$ and the estimation error is in the order of $\norm{\eVec^{pcr}}_1$
which is assumed to be small.
Hence, the NNLAD estimator with a random walk matrix of a lossless expander
might even result in less false positives and false
negatives than individual testing.\\
Note that the lack of knowledge about the noise $\eVec$ favors the
NNLAD recovery method over a \refP{Problem:BPDN} approach.
Further, since the total sum of viruses in all patients given by
$\sum_{n\in\SetOf{N}}x_n=\norm{\xVec}_1$ is unknown, it is undesirable to use \refP{Problem:CLR}.
\section*{Acknowledgments}
The work was partially supported by DAAD grant 57417688. PJ has been supported by DFG grant JU 2795/3.
BB has been supported by BMBF through the German Research Chair at AIMS, administered by the Humboldt Foundation.
\section{Appendix}\label{Section:Appendix}
\subsection{Proof of NNLR Recovery Guarantee}\label{Subsection:Proof:NNLDMinimizer}
By $\IDVec$ we denote the all ones vector in $\mathbb{R}^N$ or $\mathbb{R}^M$ respectively.
The proof is an adaption of the steps used in \cite{NNLS_first}.
As for most convex optimization problems in compressed sensing we require \cite[Theorem~4.25]{Introduction_CS}
and \cite[Theorem~4.20]{Introduction_CS} respectively.
\begin{Theorem}[ {\cite[Theorem~4.25]{Introduction_CS} \& \cite[Theorem~4.20]{Introduction_CS}} ]\label{Theorem:SRNSPCond}
	Let $q\in\left[1,\infty\right)$ and suppose $\AMat$ has the $\ell_q$-RNSP of order $S$ with respect to $\normRHS{\cdot}$ with constants
	$\rho$ and $\tau$. Then, it holds that
	\begin{align}\notag
		\norm{\xVec-\zVec}_q
		\leq \frac{\left(1+\rho\right)^2}{1-\rho}S^{\frac{1}{q}-1}
			\left(\norm{\zVec}_1-\norm{\xVec}_1+2
			d_1\left(\xVec,\Sigma_S\right)\right)
			+\frac{3+\rho}{1-\rho}\tau\normRHS{\AMat\left(\xVec-\zVec\right)}
		\TextForAll \xVec,\zVec\in\mathbb{R}^n.
	\end{align}
	If $q=1$, this bound can be improved to 
	\begin{align}\notag
		\norm{\xVec-\zVec}_1
		\leq \frac{1+\rho}{1-\rho}
			\left(\norm{\zVec}_1-\norm{\xVec}_1+2
			d_1\left(\xVec,\Sigma_S\right)\right)
			+\frac{2}{1-\rho}\tau\normRHS{\AMat\left(\xVec-\zVec\right)}
		\TextForAll \xVec,\zVec\in\mathbb{R}^n.
	\end{align}
\end{Theorem}
Note that by a modification of the proof this result also holds for $q=\infty$.
The modifications on the proofs of \cite[Theorem~4.25]{Introduction_CS} and
\cite[Theorem~4.20]{Introduction_CS} are straight forward, only
the modification of \cite[Theorem~2.5]{Introduction_CS} might not be obvious.
See also \cite{rLASSO}. As a consequence, all our statements also hold for $q=\infty$
with $\frac{1}{q}:=0$.
If $\WMat\in\mathbb{R}^{N\times N}$ is a diagonal matrix, we can calculate some operator norms fairly easy:
\begin{align}\notag
	\norm{\WMat}_{q\rightarrow q}
	:=\sup_{\norm{\wVec}_q\leq 1}\norm{\WMat\wVec}_q
	=\max_{n\in\SetOf{N}}\abs{W_{n,n}} \TextForAll q\in\left[1,\infty\right].
\end{align}
We use this relation frequently over this section.
Furthermore, we use \cite[Lemma~5]{NNLS_first} without adaption. For the sake of completeness we add a short proof.
\begin{Lemma}[ {\cite[Lemma~5]{NNLS_first}} ]\label{Lemma:DiagMat}
	Let $q\in\left[1,\infty\right)$ and suppose that $\AMat\in\mathbb{R}^{M\times N}$ has $\ell_q$-RNSP of order $S$ with respect to $\normRHS{\cdot}$ with constants
	$\rho$ and $\tau$. Let $\WMat\in\mathbb{R}^{N\times N}$ be a diagonal matrix with $W_{n,n}>0$. If
	$\rho'=\norm{\WMat}_{q\rightarrow q}\norm{\WMat^{-1}}_{1\rightarrow 1}\rho<1$,
	then $\AMat\WMat^{-1}$
	has $\ell_q$-RNSP of order $S$ with respect to $\normRHS{\cdot}$ with constants
	$\rho'=\norm{\WMat}_{q\rightarrow q}\norm{\WMat^{-1}}_{1\rightarrow 1}\rho$
	and $\tau'=\norm{\WMat}_{q\rightarrow q}\tau$.
\end{Lemma}
\begin{proof}
	Let $\vVec\in\mathbb{R}^N$ and $\SetSize{T}\leq S$. If we apply the RNSP of $\AMat$ for
	the vector $\ProjToIndex{T}{\left(\WMat^{-1}\vVec\right)}$, we get
	\begin{align}\notag
		\norm{\ProjToIndex{T}{\vVec}}_q
		=&\norm{\WMat\WMat^{-1}\left(\ProjToIndex{T}{\vVec}\right)}_q
		\leq\norm{\WMat}_{q\rightarrow q}\norm{\WMat^{-1}\left(\ProjToIndex{T}{\vVec}\right)}_q
		=\norm{\WMat}_{q\rightarrow q}\norm{\ProjToIndex{T}{\left(\WMat^{-1}\vVec\right)}}_q
		\\\notag
		\leq&\norm{\WMat}_{q\rightarrow q}\left(
				\rho S^{\frac{1}{q}-1}\norm{\ProjToIndex{T^c}{\left(\WMat^{-1}\vVec}\right)}_1
				+\tau\normRHS{\AMat\WMat^{-1}\vVec}
			\right)
		\\\notag
		=&\norm{\WMat}_{q\rightarrow q}
				\rho S^{\frac{1}{q}-1}\norm{\WMat^{-1}\left(\ProjToIndex{T^c}{\vVec}\right)}_1
			+\norm{\WMat}_{q\rightarrow q}\tau\normRHS{\AMat\WMat^{-1}\vVec}
		\\\notag
		\leq&\norm{\WMat}_{q\rightarrow q}\norm{\WMat^{-1}}_{1\rightarrow 1}
				\rho S^{\frac{1}{q}-1}\norm{\ProjToIndex{T^c}{\vVec}}_1
			+\norm{\WMat}_{q\rightarrow q}\tau\normRHS{\AMat\WMat^{-1}\vVec}.
	\end{align}
	This finishes the proof.
\end{proof}
Next we adapt \cite[Theorem~4]{NNLS_first} to account for arbitrary norms. Further, we obtain a slight improvement
in form of the dimensional scaling constant $S^{\frac{1}{q}-1}$. 
With this, our error bound becomes for $S\rightarrow\infty$ asymptotically
the error bound of the basis pursuit denoising, whenever $\kappa=1$ and $q>1$ \cite{Introduction_CS}.
\begin{Proposition}[ Similar to {\cite[Theorem~4]{NNLS_first}} ]\label{Proposition:SRNSP+M}
	Let $q\in\left[1,\infty\right)$ and $\normRHS{\cdot}$ be a norm on $\mathbb{R}^M$ with dual norm $\normDual{\cdot}$.
	Suppose $\AMat$ has $\ell_q$-RNSP of order $S$ with respect to $\normRHS{\cdot}$
	with constants $\rho$ and $\tau$.
	Suppose $\AMat$ has the $M^+$ criterion with vector $\tVec$ and constant $\kappa$
	and that $\kappa\rho<1$.
	Then, we have
	\begin{align}\notag
		\norm{\xVec-\zVec}_q
		\leq& 2\frac{\left(1+\kappa\rho\right)^2}{1-\kappa\rho}
			\kappa S^{\frac{1}{q}-1}
			d_1\left(\xVec,\Sigma_S\right)
		+\left(
				\frac{\left(1+\kappa\rho\right)^2}{1-\kappa\rho} S^{\frac{1}{q}-1}\max_{n\in\SetOf{N}}\abs{\left(\AMat^T\tVec\right)^{-1}_n}
				\normDual{\tVec}+\frac{3+\kappa\rho}{1-\kappa\rho}\kappa\tau
			\right)\normRHS{\AMat\zVec-\AMat\xVec}
		\\\notag
		&\TextForAll \xVec,\zVec\in\mathbb{R}_+^N.
	\end{align}
	If $q=1$, this bound can be improved to 
	\begin{align}\notag
		\norm{\xVec-\zVec}_q
		\leq& 2\frac{1+\kappa\rho}{1-\kappa\rho}
			\kappa 
			d_1\left(\xVec,\Sigma_S\right)
			+\left(
				\frac{1+\kappa\rho}{1-\kappa\rho} \max_{n\in\SetOf{N}}\abs{\left(\AMat^T\tVec\right)^{-1}_n}
				\normDual{\tVec}+\frac{2}{1-\kappa\rho}\kappa\tau
			\right)\normRHS{\AMat\zVec-\AMat\xVec}
		\\\notag
		& \TextForAll \xVec,\zVec\in\mathbb{R}_+^N.
	\end{align}
\end{Proposition}
\begin{proof}
	Let $\xVec,\zVec\geq 0$. In order to apply \thref{Lemma:DiagMat}
	we set $\WMat$ as the matrix with diagonal $\AMat^T\tVec$ and zero else. It follows that $W_{n,n}>0$ and
	$\norm{\WMat}_{q\rightarrow q}\norm{\WMat^{-1}}_{1\rightarrow 1}\rho=\kappa\rho<1$.
	We can apply \thref{Lemma:DiagMat}, which yields that $\AMat\WMat^{-1}$ has $\ell_q$-RNSP
	with constants $\rho'=\norm{\WMat}_{q\rightarrow q}\norm{\WMat^{-1}}_{1\rightarrow 1}\rho=\kappa\rho$
	and $\tau'=\norm{\WMat}_{q\rightarrow q}\tau=\max_{n\in\SetOf{N}}\abs{\left(\AMat^T\tVec\right)_n}\tau$.
	We apply \thref{Theorem:SRNSPCond}
	with the matrix $\AMat\WMat^{-1}$, the vectors $\WMat\xVec$, $\WMat\zVec$ and the constants
	$\rho'$ and $\tau'$ and get
	\begin{align}\notag
		\norm{\WMat\xVec-\WMat\zVec}_q
		\leq&\frac{\left(1+\rho'\right)^2}{1-\rho'}S^{\frac{1}{q}-1}
			\left(\norm{\WMat\zVec}_1-\norm{\WMat\xVec}_1+2
				d_1\left(\WMat\xVec,\Sigma_S\right)\right)
			+\frac{3+\rho'}{1-\rho'}\tau'\normRHS{\AMat\WMat^{-1}\left(\WMat\xVec-\WMat\zVec\right)}
		\\\notag
		\leq&\frac{\left(1+\rho'\right)^2}{1-\rho'}S^{\frac{1}{q}-1}
			\left(\norm{\WMat\zVec}_1-\norm{\WMat\xVec}_1+2\norm{\WMat}_{1\rightarrow 1}
				d_1\left(\xVec,\Sigma_S\right)\right)
			+\frac{3+\rho'}{1-\rho'}\tau'\normRHS{\AMat\xVec-\AMat\zVec}
		\\\notag
		=&2\frac{\left(1+\kappa\rho\right)^2}{1-\kappa\rho}
			\max_{n\in\SetOf{N}}\abs{\left(\AMat^T\tVec\right)_n}S^{\frac{1}{q}-1}
				d_1\left(\xVec,\Sigma_S\right)
			+\frac{\left(1+\kappa\rho\right)^2}{1-\kappa\rho}S^{\frac{1}{q}-1}
				\left(\norm{\WMat\zVec}_1-\norm{\WMat\xVec}_1\right)
		\\\notag
		&+\frac{3+\kappa\rho}{1-\kappa\rho}\max_{n\in\SetOf{N}}\abs{\left(\AMat^T\tVec\right)_n}\tau
				\normRHS{\AMat\xVec-\AMat\zVec}.
	\end{align}
	We lower bound the left hand side further to get
	\begin{align}\notag
		\norm{\xVec-\zVec}_q
		\leq&\norm{\WMat^{-1}}_{q\rightarrow q}\norm{\WMat\xVec-\WMat\zVec}_q		
		=\max_{n\in\SetOf{N}}\abs{\left(\AMat^T\tVec\right)^{-1}_n}\norm{\WMat\xVec-\WMat\zVec}_q
		\\\notag
		\leq&2\frac{\left(1+\kappa\rho\right)^2}{1-\kappa\rho}
			\kappa S^{\frac{1}{q}-1}
				d_1\left(\xVec,\Sigma_S\right)
			+\frac{\left(1+\kappa\rho\right)^2}{1-\kappa\rho}S^{\frac{1}{q}-1}
				\max_{n\in\SetOf{N}}\abs{\left(\AMat^T\tVec\right)^{-1}_n}
				\left(\norm{\WMat\zVec}_1-\norm{\WMat\xVec}_1\right)
		\\\label{Equation:EQ1:Proposition:SRNSP+M:Equation:Ineq1}
		&+\frac{3+\kappa\rho}{1-\kappa\rho}\kappa\tau
				\normRHS{\AMat\xVec-\AMat\zVec}.
	\end{align}
	We want to estimate the term $\norm{\WMat\zVec}_1-\norm{\WMat\xVec}_1$ using the $M^+$ criterion.
	Since $\zVec,\xVec\geq 0$, $W_{n,n}=\left(\AMat^T\tVec\right)_n>0$ and $\WMat$ is a diagonal matrix, we have
	\begin{align}\notag
		\norm{\WMat\zVec}_1-\norm{\WMat\xVec}_1
		=&\scprod{\IDVec}{\WMat\zVec}-\scprod{\IDVec}{\WMat\xVec}
		=\scprod{\WMat^T\IDVec}{\zVec-\xVec}
		=\scprod{\WMat\IDVec}{\zVec-\xVec}
		\\\notag
		=&\scprod{\tVec}{\AMat\left(\zVec-\xVec\right)}
		\leq\normDual{\tVec}\normRHS{\AMat\zVec-\AMat\xVec}.
	\end{align}
	Applying this to \refP{Equation:EQ1:Proposition:SRNSP+M:Equation:Ineq1} we get
	\begin{align}\notag
		\norm{\xVec-\zVec}_q
		\leq&2\frac{\left(1+\kappa\rho\right)^2}{1-\kappa\rho}
			\kappa S^{\frac{1}{q}-1}
				d_1\left(\xVec,\Sigma_S\right)
			+\left(
				\frac{\left(1+\kappa\rho\right)^2}{1-\kappa\rho} S^{\frac{1}{q}-1}\max_{n\in\SetOf{N}}\abs{\left(\AMat^T\tVec\right)^{-1}_n}
				\normDual{\tVec}+\frac{3+\kappa\rho}{1-\kappa\rho}\kappa\tau
			\right)\normRHS{\AMat\zVec-\AMat\xVec}.
	\end{align}
	If $q=1$ we can repeat the proof with the improved bound of \thref{Theorem:SRNSPCond}.
\end{proof}
After these auxiliary statements it remains to prove the main result of \refP{Section:main_results}
about the properties of the NNLR minimizer.
\begin{proof}[Proof of \thref{Theorem:NNLDMinimizer}]
	By applying \thref{Proposition:SRNSP+M} with $\xVec$ and $\zVec:=\xVec^\#\geq 0$ we get
	\begin{align}\notag
		\norm{\xVec-\xVec^\#}_q
		\leq&2\frac{\left(1+\kappa\rho\right)^2}{1-\kappa\rho}
			\kappa S^{\frac{1}{q}-1}
				d_1\left(\xVec,\Sigma_S\right)
			+\left(
				\frac{\left(1+\kappa\rho\right)^2}{1-\kappa\rho} S^{\frac{1}{q}-1}\max_{n\in\SetOf{N}}\abs{\left(\AMat^T\tVec\right)^{-1}_n}
				\normDual{\tVec}+\frac{3+\kappa\rho}{1-\kappa\rho}\kappa\tau
			\right)\normRHS{\AMat\xVec^\#-\AMat\xVec}
		\\\notag
		\leq&2\frac{\left(1+\kappa\rho\right)^2}{1-\kappa\rho}
			\kappa S^{\frac{1}{q}-1}
				d_1\left(\xVec,\Sigma_S\right)
		\\\notag
		&+\left(
				\frac{\left(1+\kappa\rho\right)^2}{1-\kappa\rho} S^{\frac{1}{q}-1}\max_{n\in\SetOf{N}}\abs{\left(\AMat^T\tVec\right)^{-1}_n}
				\normDual{\tVec}+\frac{3+\kappa\rho}{1-\kappa\rho}\kappa\tau
			\right)\left(\normRHS{\AMat\xVec^\#-\yVec}+\normRHS{\AMat\xVec-\yVec}\right)
		\\\notag
		\leq&2\frac{\left(1+\kappa\rho\right)^2}{1-\kappa\rho}
			\kappa S^{\frac{1}{q}-1}
				d_1\left(\xVec,\Sigma_S\right)
			+2\left(
				\frac{\left(1+\kappa\rho\right)^2}{1-\kappa\rho} S^{\frac{1}{q}-1}\max_{n\in\SetOf{N}}\abs{\left(\AMat^T\tVec\right)^{-1}_n}
				\normDual{\tVec}+\frac{3+\kappa\rho}{1-\kappa\rho}\kappa\tau
			\right)\normRHS{\AMat\xVec-\yVec},
	\end{align}
	where in the last step we used that $\xVec^\#$ is a minimizer and $\xVec$ is feasible.
	If $q=1$, we can repeat the proof with the improved bound of \thref{Proposition:SRNSP+M}.
\end{proof}
\subsection{Proof of Convergence Guarantee}\label{Subsection:proof_of_convergence_result}
We provide the exact convergence guarantee of \refP{Section:APP}
and deduce it from \cite{APP}. 
\begin{Proposition}[Convergence Guarantee]\label{Proposition:NNLAD_APP}
	Let $\AMat\in\mathbb{R}^{M\times N}$, $\yVec\in\mathbb{R}^M$.
	Further, let $\tau,\sigma\in\left(0,\infty\right)$ be parameters with $\sigma\tau<\norm{\AMat}_{2\rightarrow 2}^{-2}$
	and $\xVec^0\in\mathbb{R}^N$, $\wVec^0\in\mathbb{R}^M$ be initializations.
	Set $\vVec^0:=\xVec^0$ and for all $k\in\mathbb{N}_0$ inductively
	\begin{align}\tag{iter 1}\label{Equation:EQ1:Proposition:NNLAD_APP}
		\wVec^{k+1}:=&\left(\min\left\{1,\abs{w^{k}_m+\sigma\left(\AMat\vVec^{k}-\yVec\right)_m}\right\}
			\sgn{w^{k}_m+\sigma\left(\AMat\vVec^{k}-\yVec\right)_m}\right)_{m\in\SetOf{M}},
		\\\tag{iter 2}\label{Equation:EQ2:Proposition:NNLAD_APP}
		\xVec^{k+1}:=&\ProjToSet{\mathbb{R}^N_+}{\xVec^{k}-\tau\AMat^T\wVec^{k+1}},
		\\\tag{iter 3}\label{Equation:EQ3:Proposition:NNLAD_APP}
		\vVec^{k+1}:=&2\xVec^{k+1}-\xVec^{k},
		\\\notag
		\bar{\xVec}^{k+1}:=&\frac{1}{k+1}\sum_{k'=1}^{k+1}\xVec^{k'}
		\TextAnd
		\bar{\wVec}^{k+1}:=\frac{1}{k+1}\sum_{k'=1}^{k+1}\wVec^{k'}.
	\end{align}
	Then, the following statements hold true:
	\begin{itemize}
		\item[(1)]
			The iterates and averages converge:\\
			The sequences $\left(\xVec^k\right)_{k\in\mathbb{N}}$ and  $\left(\bar{\xVec}^k\right)_{k\in\mathbb{N}}$
			converge to a minimizer of $\argmin{\zVec\geq 0} \norm{\AMat\zVec-\yVec}_1$.
		\item[(2)]
			The iterates and averages are feasible:\\
			We have $\xVec^k\geq 0$, $\bar{\xVec}^k\geq 0$ and $\norm{\wVec^k}_\infty\leq 1$,
			$\norm{\bar{\wVec}^k}_\infty\leq 1$.
		\item[(3)]
			There is a stopping criteria for the iterates:\\
			$\lim_{k\rightarrow\infty}\norm{\AMat\xVec^k -\yVec}_1+\scprod{\yVec}{\wVec^k}=0$
			and $\lim_{k\rightarrow\infty}\AMat^T\wVec^k\geq 0$.
			In particular,
			if $\norm{\AMat\xVec^k -\yVec}_1+\scprod{\yVec}{\wVec^k}\leq 0$ and $\AMat^T\wVec^k\geq 0$,
			then $\xVec^k$ is a minimizer of $\argmin{\zVec\geq 0} \norm{\AMat\zVec-\yVec}_1$.
		\item[(4)]
			The stopping criteria also holds for the averages
			by replacing $\xVec^k$ with $\bar{\xVec}^k$ and $\wVec^k$ with $\bar{\wVec}^k$.
		\item[(5)]
			The averages obey the convergence rate to optimal objective value:\\
			$\norm{\AMat\bar{\xVec}^k-\yVec}_1-\norm{\AMat\xVec^\#-\yVec}_1
			\leq \frac{1}{k}\left(\frac{1}{2\tau}\norm{\xVec^\#-\xVec^0}_2^2
				+\frac{1}{2\sigma}\left(\norm{\wVec^0}_2^2+2\norm{\wVec^0}_1+M\right)\right)$,
			where $\xVec^\#$ is a minimizer of $\argmin{\zVec\geq 0} \norm{\AMat\zVec-\yVec}_1$.
	\end{itemize}
\end{Proposition}
In order to prove \thref{Proposition:NNLAD_APP} we introduce saddle point problems and technical notations from optimization.
Let $\f{}:\mathbb{R}^N\times\mathbb{R}^M\rightarrow\mathbb{R}\cup \left\{-\infty,\infty\right\}$.
If there exists $\left(\xVec^\#,\wVec^\#\right)$ such that
\begin{align}\notag
	\sup_{\wVec\in\mathbb{R}^M}\inf_{\xVec\in\mathbb{R}^N}\f{\xVec,\wVec}
	=\inf_{\xVec\in\mathbb{R}^N}\f{\xVec,\wVec^\#}
	=\sup_{\wVec\in\mathbb{R}^M}\f{\xVec^\#,\wVec}
	=\inf_{\xVec\in\mathbb{R}^N}\sup_{\wVec\in\mathbb{R}^M}\f{\xVec,\wVec}
\end{align}
holds true, then $\left(\xVec^\#,\wVec^\#\right)$ is called saddle point of $\f{}$.
In general we have for any point $\left(\xVec',\wVec'\right)$
\begin{align}\label{Equation:EQ1:Section:proof_of_convergence_result}
	\inf_{\xVec\in\mathbb{R}^N}\f{\xVec,\wVec'}
	\leq\f{\xVec',\wVec'}
	\leq\sup_{\wVec\in\mathbb{R}^M}\f{\xVec',\wVec}.
\end{align}
This yields that the inequality
\begin{align}\label{Equation:EQ11:Section:proof_of_convergence_result}
	\sup_{\wVec\in\mathbb{R}^M}\inf_{\xVec\in\mathbb{R}^N}\f{\xVec,\wVec}
	\leq\inf_{\xVec\in\mathbb{R}^N}\sup_{\wVec\in\mathbb{R}^M}\f{\xVec,\wVec}
\end{align}
holds true, but not necessarily with equality. The equality is a condition of the existence of a saddle point.
The problem $\inf_{\xVec\in\mathbb{R}^N}\sup_{\wVec\in\mathbb{R}^M}\f{\xVec,\wVec}$ is called the primal problem,
while the problem $\sup_{\wVec\in\mathbb{R}^M}\inf_{\xVec\in\mathbb{R}^N}\f{\xVec,\wVec}$
is called the dual problem. The difference
$\inf_{\xVec\in\mathbb{R}^N}\sup_{\wVec\in\mathbb{R}^M}\f{\xVec,\wVec}
	-\sup_{\wVec\in\mathbb{R}^M}\inf_{\xVec\in\mathbb{R}^N} \f{\xVec,\wVec}\geq 0$
is called the duality gap.
Further, \refP{Equation:EQ1:Section:proof_of_convergence_result} and \refP{Equation:EQ11:Section:proof_of_convergence_result}
yield the logical statement
\begin{align}\label{Equation:EQ2:Section:proof_of_convergence_result}
	\sup_{\wVec\in\mathbb{R}^M}\f{\xVec',\wVec}\leq\inf_{\xVec\in\mathbb{R}^N}\f{\xVec,\wVec'}
	\Rightarrow
	\left(\xVec',\wVec'\right) \text{ is a saddle point.}
\end{align}
Given a function $\F{}:\mathbb{R}^N\rightarrow\mathbb{R}\cup\left\{-\infty,\infty\right\}$
its Fenchel conjugate is the function $\Fstar{}:\mathbb{R}^N\rightarrow\left[-\infty,\infty\right]$,
where $\Fstar{\vVec}:=\sup_{\vVec^\ast\in\mathbb{R}^N} \scprod{\vVec}{\vVec^\ast}-\F{\vVec^\ast}$.
The fenchel conjugate has several interesting properties, however we only require that
if $\F{}$ is proper, convex and lower semicontinuous
\footnote{Note that in general convex and lower semicontinuous need to be defined with the epigraph,
	since we allow $\F{}$ to attain the values $-\infty$ and $\infty$, which might result in
	undefined $\infty-\infty$ terms.
	However, if $\F{}$ is proper as in our case, it can only attain $\infty$ and thus
	the casual definitions of algebra coincide with the definitions used here.}
, then also $\Fstar{}$ is proper, convex and lower semicontinuous
and ${\Fstar{}}^\ast=\F{}$ holds true \cite[Theorem~12.2]{book_convex}.
Given a proper, convex, lower-semicontinuous function $\F{}:\mathbb{R}^N\rightarrow\mathbb{R}$,
the proximal point operator of $\F{}$ is the function
$\Prox{\F{}}{\cdot}:\mathbb{R}^N\rightarrow\mathbb{R}^N$, where
$\Prox{\F{}}{\vVec}$ is the unique minimizer of
\begin{align}\notag
		\argmin{\vVec^\ast\in\mathbb{R}^N}\frac{1}{2}\norm{\vVec^\ast-\vVec}_2^2+\F{\vVec^\ast}
\end{align}
\cite[Theorem~31.5]{book_convex}.
For more information about saddle point problems, the fenchel conjugate and proximal point operators
we refer the reader to \cite{book_convex}.
We have now the necessary means to state \cite[Theorem~1]{APP}.\\
\begin{Theorem}[ {\cite[Theorem~1]{APP}} ]\label{Theorem:APP}
	Let $\F{}:\mathbb{R}^M\rightarrow\left[0,\infty\right)$ be convex and lower semicontinuous.
	Let $\G{}:\mathbb{R}^N\rightarrow\left[0,\infty\right]$ and $\Fstar{}:\mathbb{R}^M\rightarrow\left[0,\infty\right)$
	\footnote{Note that the result in \cite{APP} is only stated if $\F{},\G{}$ map to $\left[0,\infty\right)$.
		From a private conversation with one of the authors we
		learned that the result also holds if $\G{}$ maps to $\left[0,\infty\right]$.}
	be proper, convex and lower semicontinous functions
	and $\AMat\in\mathbb{R}^{M\times N}$.
	Then, the function $\f{\xVec,\wVec}:=\scprod{\AMat\xVec}{\wVec}+\G{\xVec}-\Fstar{\wVec}$
	has a saddle point.\\
	Further, let $\tau,\sigma\in\left(0,\infty\right)$ be parameters with $\sigma\tau<\norm{\AMat}_{2\rightarrow 2}^{-2}$
	and $\xVec^0\in\mathbb{R}^N$, $\wVec^0\in\mathbb{R}^M$ be initializations.
	Set $\vVec^0:=\xVec^0$ and for all $k\in\mathbb{N}_0$ inductively
	\begin{align}\tag{PP 1}\label{Equation:EQ1:Theorem:APP}
		\wVec^{k+1}=&\Prox{\sigma\Fstar{}}{\wVec^{k}+\sigma\AMat\vVec^{k}}
		\\\tag{PP 2}\label{Equation:EQ2:Theorem:APP}
		\xVec^{k+1}=&\Prox{\tau\G{}}{\xVec^{k}-\tau\AMat^T\wVec^{k+1}}
		\\\tag{PP 3}\label{Equation:EQ3:Theorem:APP}
		\vVec^{k+1}=&2\xVec^{k+1}-\xVec^{k}.
	\end{align}
	The sequence $\left(\xVec^k,\wVec^k\right)$ converges to a saddle point	of $\f{}$.
	Lastly, for any bounded sets $B_1\subset\mathbb{R}^N$ and $B_2\subset\mathbb{R}^M$ the averages
	$\left(\bar{\xVec}^k,\bar{\wVec}^k\right)
		:=\left(\frac{1}{k}\sum_{k'=1}^k\xVec^{k'},\frac{1}{k}\sum_{k'=1}^k\wVec^{k'}\right)$ obey
	\begin{align}\notag
		\sup_{\wVec\in B_2}\f{\bar{\xVec}^k,\wVec}
			-\inf_{\xVec\in B_1}\f{\xVec,\bar{\wVec}^k}
		\leq\frac{1}{k}\sup_{\xVec\in B_1,\wVec\in B_2}\left(\frac{1}{2\tau}\norm{\xVec-\xVec^0}_2^2+\frac{1}{2\sigma}\norm{\wVec-\wVec^0}_2^2\right).
	\end{align}
\end{Theorem}
By a proper choice of $\F{}$ and $\G{}$ any saddle point of $\f{}$ will also give a minimizer of NNLAD. We denote this proper choice in the next lemma.
\begin{Lemma}[ Relation of Saddle point and NNLAD ]\label{Lemma:Saddle_to_Min}
	Let $\yVec\in\mathbb{R}^M$ as well as
	\begin{align}\notag
		\F{\wVec}:=\norm{\wVec-\yVec}_1
		\TextAnd \G{\xVec}
		:=\begin{Bmatrix}
			0 & \TextIf & \xVec\geq 0
			\\
			\infty & \TextElse & 
		\end{Bmatrix}
		\TextAnd \f{\xVec,\wVec}:=\scprod{\AMat\xVec}{\wVec}+\G{\xVec}-\Fstar{\wVec}.
	\end{align}
	Then $\F{},\G{},\Fstar{},\Gstar{}$ are proper, convex and lower semicontinuous. $\Fstar{}$ and $\Gstar{}$ are given by
	\begin{align}\notag
		\Fstar{\wVec}
		=\begin{Bmatrix}
			\scprod{\wVec}{\yVec} & \TextIf & \norm{\wVec}_\infty\leq 1
			\\
			\infty & \TextIf & \norm{\wVec}_\infty>1
		\end{Bmatrix}
		\TextAnd\Gstar{\xVec}
		=\begin{Bmatrix}
			0 & \TextIf & \xVec\leq 0
			\\
			\infty & \TextElse &
		\end{Bmatrix}.
	\end{align}
	Further we have for $\xVec'\in\mathbb{R}^N$, $\wVec'\in\mathbb{R}^M$
	\begin{align}\notag
		\sup_{\wVec\in\mathbb{R}^M}\f{\xVec',\wVec}
		=&\begin{Bmatrix}
			\norm{\AMat\xVec'-\yVec}_1 & \TextIf & \xVec'\geq 0
			\\
			\infty & \TextElse &
		\end{Bmatrix}
		\\\notag
		\TextAnd\inf_{\xVec\in\mathbb{R}^N}\f{\xVec,\wVec'}
		=&\begin{Bmatrix}
			-\scprod{\wVec'}{\yVec} & \TextIf & \norm{\wVec'}_\infty\leq 1 \TextAnd \AMat^T\wVec'\geq 0
			\\
			-\infty & \TextElse & 
		\end{Bmatrix}.
	\end{align}
\end{Lemma}
\begin{proof}
	From the definition it is clear that $\F{},\G{}$ are proper, convex and lower semicontinuous.
	Hence $\Fstar{}$ and $\Gstar{}$ are also proper, convex and lower semicontinuous.
	By a direct calculation we have
	\begin{align}\notag
		\Gstar{\xVec}
		=\sup_{\xVec^\ast\in\mathbb{R}^N}\scprod{\xVec}{\xVec^\ast}-\G{\xVec^\ast}
		=\sup_{\xVec^\ast\geq 0}\scprod{\xVec}{\xVec^\ast}
		=\begin{Bmatrix}
			0 & \TextIf & \xVec\leq 0
			\\
			\infty & \TextElse &
		\end{Bmatrix}.
	\end{align}
	For the other fenchel conjugate we calculate
	\begin{align}\notag
		\Fstar{\wVec}
		=&\sup_{\wVec^\ast\in\mathbb{R}^M}\scprod{\wVec}{\wVec^\ast}-\norm{\wVec^\ast-\yVec}_1
		=\scprod{\wVec}{\yVec}+\sup_{\wVec^\ast\in\mathbb{R}^M}\scprod{\wVec}{\wVec^\ast}-\norm{\wVec^\ast}_1
		\\\notag
		=&\scprod{\wVec}{\yVec}+
			\sup_{\wVec^\ast\in\mathbb{R}^M}\sum_{m\in\SetOf{M}}w_m w^\ast_m-\abs{w^\ast_m}
		=\scprod{\wVec}{\yVec}+\sum_{m\in\SetOf{M}}
			\sup_{w^\ast\in\mathbb{R}}w_m w^\ast-\abs{w^\ast},
	\end{align}
	where in the last step we used that each summand depends on exactly one component of $\wVec^\ast$.
	Now $w_m w^\ast-\abs{w^\ast}$ is larger for $\sgn{w^\ast}=\sgn{w_m}$, than for $\sgn{w^\ast}\neq\sgn{w_m}$. Hence,
	we can restrict the supremum to the case $\sgn{w^\ast}=\sgn{w_m}$ and obtain
	\begin{align}
		\notag
		\Fstar{\wVec}
		=&\scprod{\wVec}{\yVec}+\sum_{m\in\SetOf{M}}
			\sup_{w^\ast\in\mathbb{R}_+}\left(\abs{w_m}-1\right)w^\ast
		=\scprod{\wVec}{\yVec}+\sum_{m\in\SetOf{M}}
			\begin{Bmatrix}
				0 & \TextIf & \abs{w_m}\leq 1
				\\
				\infty & \TextIf & \abs{w_m}>1
			\end{Bmatrix}
		\\\notag
		=&\begin{Bmatrix}
				\scprod{\wVec}{\yVec} & \TextIf & \norm{\wVec}_\infty\leq 1
				\\
				\infty & \TextIf & \norm{\wVec}_\infty>1
			\end{Bmatrix}.
	\end{align}
	Since $\F{}$ is proper, convex and lower semicontinuous, we have ${\Fstar{}}^\ast=\F{}$. Thus,
	\begin{align}\notag
		\sup_{\wVec\in\mathbb{R}^M}\f{\xVec',\wVec}
		=\G{\xVec'}+\sup_{\wVec\in\mathbb{R}^M}\scprod{\AMat\xVec'}{\wVec}-\Fstar{\wVec}
		=\G{\xVec'}+\F{\AMat\xVec'}
		=\begin{Bmatrix}
			\norm{\AMat\xVec'-\yVec}_1 & \TextIf & \xVec'\geq 0
			\\
			\infty & \TextElse &
		\end{Bmatrix}.
	\end{align}
	And lastly we have
	\begin{align}\notag
		\inf_{\xVec\in\mathbb{R}^N}\f{\xVec,\wVec'}
		=&-\Fstar{\wVec'}+\inf_{\xVec\in\mathbb{R}^N}\scprod{\xVec}{\AMat^T\wVec'}+\G{\xVec}
		=-\Fstar{\wVec'}-\sup_{\xVec\in\mathbb{R}^N}\scprod{\xVec}{-\AMat^T\wVec'}-\G{\xVec}
		\\\notag
		=&-\Fstar{\wVec'}-\Gstar{-\AMat^T\wVec'}
		=\begin{Bmatrix}
			-\scprod{\wVec'}{\yVec} & \TextIf & \norm{\wVec'}_\infty\leq 1 \TextAnd \AMat^T\wVec'\geq 0
			\\
			-\infty & \TextElse & 
		\end{Bmatrix},
	\end{align}
	which finishes the proof.	
\end{proof}
Further, we need to calculate the iterates for this choice of $\F{}$ and $\G{}$ and thus the proximal point operators.
It is well known that the proximal point operator of the $\ell_1$-norm is the soft thresholding operator.
Using Moreau's identity \cite[Theorem~31.5]{book_convex} one can find the desired iterates directly. See for instance
\cite[Example~15.7]{Introduction_CS}. For the sake of completeness we added a proof.
\begin{Lemma}\label{Lemma:Prox_Op}
	Let $\tau,\sigma>0$, $\yVec\in\mathbb{R}^M$ as well as
	\begin{align}\notag
		\F{\wVec}:=\norm{\wVec-\yVec}_1
		\TextAnd \G{\xVec}
		:=\begin{Bmatrix}
			0 & \TextIf & \xVec\in\mathbb{R}^N_+
			\\
			\infty & \TextElse & 
		\end{Bmatrix}.
	\end{align}
	Then
	\begin{align}\label{Equation:Lemma:Prox_Op:EQ1}
		\Prox{\sigma\F{}}{\wVec}
		=\left(\begin{Bmatrix}
			y_m & \TextIf & \abs{w_m-y_m}\leq \sigma
			\\w_m-\sigma & \TextIf & w_m-y_m>\sigma
			\\w_m+\sigma & \TextIf & w_m-y_m<-\sigma
		\end{Bmatrix}\right)_{m\in\SetOf{M}}
		\TextForAll\wVec\in\mathbb{R}^M
	\end{align}
	and in particular
	\begin{align}
		\Prox{\tau\G{}}{\xVec}
		=\ProjToSet{\mathbb{R}^N_+}{\xVec}
		\TextAnd
		\Prox{\sigma\Fstar{}}{\wVec}
		=\left(\min\left\{1,\abs{w_m-\sigma y_m}\right\}\sgn{w_m-\sigma y_m}\right)_{m\in\SetOf{M}}.
		\label{Equation:Lemma:Prox_Op:EQ2}
	\end{align}
\end{Lemma}
\begin{proof}
	The proximal point operator of an indicator function of a closed, convex set is always the projection to the set,
	hence the identity for $\G{}$ follows.
	For $\F{}$ this is more difficult. Note that $\wVec'$ is a minimizer of
	\begin{align}\notag
		\argmin{\wVec^\ast\in\mathbb{R}^M}\frac{1}{2}\norm{\wVec^\ast-\wVec}_2^2+\sigma\norm{\wVec^\ast-\yVec}_1
	\end{align}
	if and only if zero is in the subdifferential at $\wVec'$, which is given by the set
	\begin{align}\notag
		\left\{\wVec'-\wVec+\sigma \tilde{\wVec}\in\mathbb{R}^M\TextSuchThat
			\begin{Bmatrix}
				\tilde{w}_m=\sgn{w'_m-y_m} & \TextIf & w'_m-y_m\neq 0
				\\\tilde{w}_m\in\left[-1,1\right] & \TextIf & w'_m-y_m=0
			\end{Bmatrix}
		\right\}.
	\end{align}
	Since the minimizer for the proximal operator is always unique,
	it remains to verify that zero is in the subdifferential at the vector from the statement.
	So let $\wVec'$ be the vector from the right hand side of \refP{Equation:Lemma:Prox_Op:EQ1}
	and $m\in\SetOf{M}$.
	If $w_m-y_m>\sigma$, then
	$w'_m=w_m-\sigma>y_m$ and thus
	\begin{align}\notag
		\left(w'_m-w_m\right)+\sigma\sgn{w'_m-y_m}
		=-\sigma+\sigma=0.
	\end{align}
	If $w_m-y_m<-\sigma$, then
	$w'_m=w_m+\sigma<y_m$ and thus
	\begin{align}\notag
		\left(w'_m-w_m\right)+\sigma\sgn{w'_m-y_m}
		=\sigma-\sigma=0.
	\end{align}
	If $\abs{w_m-y_m}\leq \sigma$,
	we have $w'_m=y_m$ and
	\begin{align}\notag
		\abs{\left(w'_m-w_m\right)}=\abs{y_m-w_m}\leq \sigma,
	\end{align}
	and hence $\sigma^{-1}\left(w'_m-w_m\right)\in\left[-1,1\right]$.
	It follows that zero is a possible subgradient, i.e. the subdifferential contains zero.
	Hence, $\wVec'$ is the unique minimizer.
	To prove \refP{Equation:Lemma:Prox_Op:EQ2}, we apply the first statement to calculate
	\begin{align}\notag
		\Prox{\sigma^{-1}\F{}}{\sigma^{-1}\wVec}
		=\left(\begin{Bmatrix}
			y_m & \TextIf & \abs{w_m-\sigma y_m}\leq 1
			\\\sigma^{-1}\left(w_m-1\right) & \TextIf & w_m-\sigma y_m>1
			\\\sigma^{-1}\left(w_m+1\right)& \TextIf & w_m-\sigma y_m<-1
		\end{Bmatrix}\right)_{m\in\SetOf{M}}.
	\end{align}
	It follows that
	\begin{align}\notag
		\wVec-\sigma\Prox{\sigma^{-1}\F{}}{\sigma^{-1}\wVec}
		=&\left(\begin{Bmatrix}
			w_m-\sigma y_m & \TextIf & \abs{w_m-\sigma y_m}\leq 1
			\\ 1 & \TextIf & w_m-\sigma y_m>1
			\\ -1 & \TextIf & w_m-\sigma y_m<-1
		\end{Bmatrix}\right)_{m\in\SetOf{M}}
		\\\notag
		=&\left(\min\left\{1,\abs{w_m-\sigma y_m}\right\}\sgn{w_m-\sigma y_m}\right)_{m\in\SetOf{M}}.
	\end{align}
	Using Moreau's identity \cite[Theorem~31.5]{book_convex} yields
	\begin{align}\notag
		\Prox{\sigma\Fstar{}}{\wVec}
		=&\wVec-\Prox{\left(\sigma\Fstar{}\right)^\ast}{\wVec}
		=\wVec-\sigma\Prox{\sigma^{-1}\F{}}{\sigma^{-1}\wVec}
		\\\notag
		=&\left(\min\left\{1,\abs{\wVec_m-\sigma\yVec_m}\right\}\sgn{\wVec_m-\sigma\yVec_m}\right)_{m\in\SetOf{M}},
	\end{align}
	which finishes the proof.
\end{proof}
After proving these auxiliary statements it remains to prove the main result of \refP{Section:APP} about
the convergence to a minimizer of NNLAD.
\begin{proof}[Proof of \thref{Proposition:NNLAD_APP}]
	We set
	\begin{align}\notag
		\F{\wVec}:=\norm{\wVec-\yVec}_1
		\TextAnd \G{\xVec}
		:=\begin{Bmatrix}
			0 & \TextIf & \xVec\geq 0
			\\
			\infty & \TextElse & 
		\end{Bmatrix}
		\TextAnd \f{\xVec,\wVec}:=\scprod{\AMat\xVec}{\wVec}+\G{\xVec}-\Fstar{\wVec}.
	\end{align}
	By \thref{Lemma:Saddle_to_Min} $\F{},\G{},\Fstar{},\Gstar{}$ are proper, convex and lower-semicontinuous.
	Thus, the requirements of \thref{Theorem:APP} are fulfilled, which yields that
	$\f{}$ has a saddle point and thus the duality gap is zero.
	By \thref{Lemma:Saddle_to_Min} and the fact that the duality gap is zero,
	it follows that
	\begin{align}\label{Equation:saddle:Proposition:NNLAD_APP}
		\left(\xVec^\#,\wVec^\#\right) \text{is a saddle point}
		\Leftrightarrow
		\xVec^\# \in\argmin{\xVec\geq 0}\norm{\AMat\xVec-\yVec}_1
		\TextAnd\wVec^\#\in\argmax{\wVec\in\mathbb{R}^M:\AMat^T\wVec\geq 0,\norm{\wVec}_\infty\leq 1}
			-\scprod{\wVec}{\yVec}.
	\end{align}
	If $\left(\xVec',\wVec'\right)$ are any points with
	$\xVec'\geq 0$, $\norm{\wVec'}_\infty\leq 1$ and $\AMat^T\wVec'\geq 0$, then we have
	by \thref{Lemma:Saddle_to_Min}
	\begin{align}\notag
		\norm{\AMat\xVec'-\yVec}_1+\scprod{\yVec}{\wVec'}
		=\sup_{\wVec\in\mathbb{R}^M}\f{\xVec',\wVec}
			-\inf_{\xVec\in\mathbb{R}^N}\f{\xVec,\wVec'}.
	\end{align}
	If this is non-positive, \refP{Equation:EQ2:Section:proof_of_convergence_result} and
	\refP{Equation:saddle:Proposition:NNLAD_APP} yield that $\xVec'$ is a minimizer of NNLAD.
	Hence, it holds true that
	\begin{align}\label{Equation:stop:Proposition:NNLAD_APP}
		\norm{\AMat\xVec'-\yVec}_1+\scprod{\yVec}{\wVec'}\leq 0
		\TextAnd\xVec'\geq 0
		\TextAnd\norm{\wVec'}_\infty\leq 1
		\TextAnd\AMat^T\wVec'\geq 0
		\Rightarrow
		\xVec' \text{ is minimizer of NNLAD}.
	\end{align}
	Lastly, by \thref{Lemma:Prox_Op} the iterates calculated in
	\refP{Equation:EQ1:Proposition:NNLAD_APP}, \refP{Equation:EQ2:Proposition:NNLAD_APP} and
	\refP{Equation:EQ3:Proposition:NNLAD_APP} are exactly the iterates calculated in
	\refP{Equation:EQ1:Theorem:APP}, \refP{Equation:EQ2:Theorem:APP} and \refP{Equation:EQ3:Theorem:APP} respectively.
	We will now prove all statements.\\
	By \thref{Theorem:APP} the sequence $\left(\xVec^k,\wVec^k\right)$ converges to some saddle point
	$\left(\xVec^\#,\wVec^\#\right)$. Hence, $\xVec^k$ converges to $\xVec^\#$, which is a minimizer
	of NNLAD by \refP{Equation:saddle:Proposition:NNLAD_APP}.
	Since any sequence of averages converges to the same value
	as the original sequence, statement (1) follows.\\
	Since $\xVec^k$ and $\wVec^k$ are in the image of the proximal point operator of $\G{}$ and $\Fstar{}$ respectively,
	they need to obey $\G{\xVec^k}<\infty$ and $\Fstar{\wVec^k}<\infty$. \thref{Lemma:Saddle_to_Min}
	yields the $\xVec^k\geq 0$ and $\norm{\wVec^k}_\infty\leq 1$.
	By convexity we obtain also $\bar{\xVec}^k\geq 0$ and $\norm{\bar{\wVec}^k}_\infty\leq 1$.
	Statement (2) is proven.\\
	By \thref{Theorem:APP} the sequence $\left(\xVec^k,\wVec^k\right)$ converges to some saddle point
	$\left(\xVec^\#,\wVec^\#\right)$. By taking the limit, statement (2) yields
	$\xVec^\#\geq 0$ and $\norm{\wVec^\#}_\infty\leq 1$.
	The saddle point property and \thref{Lemma:Saddle_to_Min} implies
	\begin{align}\notag
		\inf_{\xVec\in\mathbb{R}^N}\f{\xVec,\wVec^\#}
		=\sup_{\wVec\in\mathbb{R}^M}\f{\xVec^\#,\wVec}
		=\norm{\AMat\xVec^\#-\yVec}_1<\infty.
	\end{align}
	By \thref{Lemma:Saddle_to_Min} again, this is only possible if $\wVec^\#$ is feasible, i.e.
	\begin{align}\label{Equation:EQ4:Proposition:NNLAD_APP}
		\AMat^T\wVec^\#\geq 0.
	\end{align}
	Hence, $\lim_{k\rightarrow\infty}\AMat^T\wVec^k\geq 0$ follows.
	By \thref{Lemma:Saddle_to_Min} and the feasibility of $\xVec^\#$ and $\wVec^\#$ we have
	\begin{align}\notag
		\lim_{k\rightarrow\infty}\norm{\AMat\xVec^k -\yVec}_1+\scprod{\yVec}{\wVec^k}
		=\norm{\AMat\xVec^\#-\yVec}_1+\scprod{\yVec}{\wVec^\#}
		=\sup_{\wVec\in\mathbb{R}^M}\f{\xVec^\#,\wVec}
			-\inf_{\xVec\in\mathbb{R}^N}\f{\xVec,\wVec^\#}
	\end{align}
	which is zero, since $\left(\xVec^\#,\wVec^\#\right)$ is a saddle point. This yields the convergence in statement (3).
	Since any sequence of averages converges to the same value
	as the original sequence, we also get the convergence of statement (4).
	The in particular part of statements (3) and (4) follows from \refP{Equation:stop:Proposition:NNLAD_APP} and statement (2).
	Hence, statements (3) and (4) are proven.\\
	To prove the the remaining statement (5) we choose
	$B_1:=\left\{\xVec^\#\right\}$ and $B_2:=\left\{\wVec:\norm{\wVec}_\infty\leq 1 \right\}$.
	The bound of \thref{Theorem:APP} becomes
	\begin{align}\notag
		\sup_{\wVec\in B_2}\f{\bar{\xVec}^k,\wVec}
			-\f{\xVec^\#,\bar{\wVec}^k}
		\leq&\frac{1}{k}\left(\frac{1}{2\sigma}\norm{\xVec^\#-\xVec^0}_2^2
			+\frac{1}{2\tau}\sup_{\norm{\wVec}_\infty\leq 1}\norm{\wVec-\wVec^0}_2^2\right)
		\\\label{Equation:EQ5:Proposition:NNLAD_APP}
		=&\frac{1}{k}\left(\frac{1}{2\sigma}\norm{\xVec^\#-\xVec^0}_2^2
			+\frac{1}{2\tau}\left(\norm{\wVec^0}_2^2+2\norm{\wVec^0}_1+M\right)\right).
	\end{align}
	By using \thref{Lemma:Saddle_to_Min} and the feasibility of $\xVec^\#$ we get
	\begin{align}\label{Equation:EQ6:Proposition:NNLAD_APP}
		\f{\xVec^\#,\bar{\wVec}^k}
		\leq\sup_{\wVec\in\mathbb{R}^M}\f{\xVec^\#,\wVec}
		=\norm{\AMat\xVec^\#-\yVec}_1.
	\end{align}
	Now let $\tilde{\wVec}$ be a maximizer of $\sup_{\wVec\in \mathbb{R}^M}\f{\bar{\xVec}^k,\wVec}$.
	By statement (2) we get $\G{\bar{\xVec}^k}=0$ and thus
	$\tilde{\wVec}$ is also a minimizer of the convex function
	$\wVec\rightarrow\scprod{-\AMat\bar{\xVec}^k}{\wVec}+\Fstar{\wVec}$.
	Hence, the subdifferential of this function needs to contain zero at $\tilde{\wVec}$.
	Since $\partial\Fstar{\wVec}=\emptyset$ whenever $\norm{\wVec}_\infty>1$, we get
	$\norm{\tilde{\wVec}}_\infty\leq 1$.
	This together with the feasibility of $\bar{\xVec}^k$ yields
	\begin{align}\label{Equation:EQ7:Proposition:NNLAD_APP}
		\sup_{\wVec\in B_2}\f{\bar{\xVec}^k,\wVec}
		=\f{\bar{\xVec}^k,\tilde{\wVec}}
		=\sup_{\wVec\in \mathbb{R}^M}\f{\bar{\xVec}^k,\wVec}
		=\norm{\AMat\bar{\xVec}^k-\yVec}_1.
	\end{align}
	Combining \refP{Equation:EQ5:Proposition:NNLAD_APP}, \refP{Equation:EQ6:Proposition:NNLAD_APP} and
	\refP{Equation:EQ7:Proposition:NNLAD_APP} yields
	\begin{align}\notag
		\norm{\AMat\bar{\xVec}^k-\yVec}_1-\norm{\AMat\xVec^\#-\yVec}_1
		\leq\frac{1}{k}\left(\frac{1}{2\sigma}\norm{\xVec^\#-\xVec^0}_2^2
			+\frac{1}{2\tau}\left(\norm{\wVec^0}_2^2+2\norm{\wVec^0}_1+M\right)\right)
	\end{align}
	and finishes the proof.
\end{proof}
We want to remark that the other feasibility assumptions $\AMat^T\wVec^k\geq 0$ and
$\AMat^T\bar{\wVec}^k\geq 0$ does not need to hold.
	\begin{small}
		\bibliography{Bibliography/Bibliography}
		\bibliographystyle{alphaurl}
	\end{small}
\end{document}